\documentclass[preprint,12pt]{elsarticle}

\usepackage{amssymb}
\usepackage{amsmath}

\usepackage{comment}
\usepackage{amsthm}

\usepackage{float}    

\usepackage{booktabs}
\usepackage{adjustbox}
\usepackage{makecell}
\usepackage{longtable}
\usepackage{rotating}

\numberwithin{equation}{section} 
\allowdisplaybreaks

\newtheorem{definition}{Definition}[section]
\newtheorem{proposition}{Proposition}[section]
\newtheorem{lemma}{Lemma}[section]
\newtheorem{theorem}{Theorem}[section]
\newtheorem{corollary}{Corollary}[section]
\newtheorem*{remark}{Remark}

\newcommand{\trace}{{\rm{Tr}}_{\mathbb{F}_{q^2}/\mathbb{F}_3}}
\newcommand{\traced}{{\rm{Tr}}_{\mathbb{F}_{q^2}/\mathbb{F}_q}}

\newcommand{\norm}{{\rm{N}}_{\mathbb{F}_{q^2}/\mathbb{F}_q}}

\journal{Arxiv}

\begin{document}

\begin{frontmatter}



  \title{An in-depth study of the power function \(x^{q+2}\) over the finite field \(\mathbb{F}_{q^2}\): the differential, boomerang, and Walsh spectra, with an application to coding theory} 


  \author[1,2,3]{Sihem Mesnager}
  \ead{smesnager@univ-paris8.fr}

  \author[1,2]{Huawei Wu\corref{cor1}}
  \ead{wuhuawei1996@gmail.com}
  \cortext[cor1]{Corresponding author}

  \affiliation[1]{organization={Department of Mathematics},
    addressline={University of Paris VIII},
    city={F-93526 Saint-Denis},
    country={France}}

  \affiliation[2]{organization={Laboratory Analysis, Geometry and Applications, LAGA, University Sorbonne Paris Nord, CNRS, UMR 7539},
    addressline={F-93430 Villetaneuse},
    country={France}}

  \affiliation[3]{organization={Telecom Paris, Polytechnic Institute of Paris},
    addressline={91120 Palaiseau},
    country={France}}


  \begin{abstract}
    
Let \( q = p^m \), where \( p \) is an odd prime number and \( m \) is a positive integer. In this paper, we examine the finite field \( \mathbb{F}_{q^2} \), which consists of \( q^2 \) elements. We first present an alternative method to determine the differential spectrum of the power function \( f(x) = x^{q+2} \) on \( \mathbb{F}_{q^2} \), incorporating several key simplifications. This methodology provides a new proof of the results established by Man, Xia, Li, and Helleseth in Finite Fields and Their Applications 84 (2022), 102100, which not only completely determine the differential spectrum of $f$ but also facilitate the analysis of its boomerang uniformity.

Specifically, we determine the boomerang uniformity of \( f \) for the cases where \( q \equiv 1 \) or \( 3 \) (mod \( 6 \)), with the exception of the scenario where \( p = 5 \) and \( m \) is even. Furthermore, for \( p = 3 \), we investigate the value distribution of the Walsh spectrum of \( f \), demonstrating that it takes on only four distinct values. Using this result, we derive the weight distribution of a ternary cyclic code with four Hamming weights.  The article integrates refined mathematical techniques from algebraic number theory and the theory of finite fields, employing several ingredients, such as exponential sums, to explore the cryptographic analysis of functions over finite fields. They can be used to explore the differential/boomerang uniformity across a wider range of functions.
  \end{abstract}


\end{frontmatter}



\section{Introduction}

Let \(\mathbb{F}_{q}\) be the finite field with \(q\) elements, where \(q = p^m\), with \(p\) as a prime number and \(m\) as a positive integer. For any function \(f: \mathbb{F}_q \rightarrow \mathbb{F}_q\) and any element \(a \in \mathbb{F}_q\), we define the derivative of \(f\) at \(a\) as 
\[
D_a f(x) = f(x + a) - f(x), \quad x \in \mathbb{F}_q.
\]
For any \(a, b \in \mathbb{F}_q\), we use \(\delta_f(a, b)\) to denote the number of preimages of \(b\) under \(D_a f\), i.e., 
\[
\delta_f(a, b) = \#\{ x \in \mathbb{F}_q:\ D_a f(x) = b \}.
\]
The differential uniformity of \(f\) is defined as 
\[
\delta_f = \max\limits_{\substack{a \in \mathbb{F}_q^* \\ b \in \mathbb{F}_q}} \delta_f(a, b),
\]
which measures \(f\)'s ability to resist differential attacks when used as a substitution box (S-box) in a cipher. A smaller differential uniformity indicates stronger resistance of the corresponding S-box. Functions that achieve the minimum differential uniformity of \(1\) are known as perfect nonlinear (PN) functions, which only exist in fields with odd characteristics. Functions with a differential uniformity of \(2\) are called almost perfect nonlinear (APN) functions, which represent the minimum value for fields with even characteristics. For further properties and applications of PN and APN functions, readers can refer to \cite{carlet2021boolean}, \cite{coulter1999class} and \cite{wu2016boolean}.

When studying the differential properties of a function $f$, knowing its differential uniformity alone often does not suffice; we also want to know the specific distribution of the values $\delta_f(a,b)$ $(a\in\mathbb{F}_q^*$, $b\in\mathbb{F}_q$). For any $0\le i\le\delta_f$, let
$$\omega_i=\#\{(a,b)\in\mathbb{F}_q^*\times\mathbb{F}_q:\ \delta_f(a,b)=i\}.$$
The differential spectrum of $f$ is defined as the following multiset
$${\rm{DS}}_f=\{\omega_i:\ 0\le i\le\delta_f\}.$$
The differential spectrum of a nonlinear function is not only crucial in symmetric cryptography but also finds broad applications in sequences \cite{dobbertin2001ternary}, coding theory \cite{blondeau2010differential, charpin2019differential}, and combinatorial design theory \cite{tang2019codes}. 

Power functions with low differential uniformity are excellent candidates for designing S-boxes due to their strong resistance to differential attacks and typically low hardware implementation costs. If we consider a function \( f(x) = x^d \) for some integer \( d \), it becomes evident that \( \delta_f(a,b) = \delta_f(1,\frac{b}{a^d}) \) for any \( a \in \mathbb{F}_q^* \) and \( b \in \mathbb{F}_q \). This means that to study the differential properties of \( f \), we only need to focus on the values \( \delta_f(1,b) \) where \( b \in \mathbb{F}_q \). Thus, for a power function \( f \) defined over \( \mathbb{F}_q \), we can define its differential spectrum as:
\[
{\rm{DS}}_f = \{\omega_i:\ 0 \leq i \leq \delta_f\},
\]
where \( \omega_i = \#\{b \in \mathbb{F}_q:\ \delta_f(1,b) = i\} \). In the subsequent section, we will adhere to this definition. We have the following fundamental property of the differential spectrum (see \cite{blondeau2010differential}):
\begin{equation}\label{20240622equation1}
  \sum\limits_{i=0}^{\delta_f}\omega_i=\sum\limits_{i=0}^{\delta_f}i\omega_i=q.
\end{equation}

A power function \( f \) over \( \mathbb{F}_q \) is said to be locally-APN (locally almost perfect nonlinear) if
\[
\max\{\delta_f(1,b):\ b \in \mathbb{F}_{q} \setminus \mathbb{F}_p\} = 2.
\]
Blondeau and Nyberg introduced the concept of locally-APNness for the case where \( p = 2 \). They demonstrated that a locally-APN S-box could produce smaller differential probabilities than others with a differential uniformity of $4$, using a cryptographic toy example \cite{blondeau2015perfect}. 

Determining the differential spectrum of a power function can be challenging. Significant research has been conducted on this topic, which we summarize in Table \ref{table1}.

\begin{table}[h!]
  \centering
  \caption{Power functions $f(x)=x^d$ over $\mathbb{F}_{p^n}$ with known differential spectrum where $p$ is odd\ (quotes in the table indicate omitted content due to length)\\}
  \label{table1}
  \begin{tabular}{llllc}
    \toprule
    $p$     & $d$                                   & Condition                         & $\delta_f$                    & References                                                               \\
    \midrule
    $3$     & $2\cdot 3^{\frac{n-1}{2}}+1$          & $n$ odd $>1$                      & $4$                            & \cite{dobbertin2001ternary}

    \\
    \hline
    $3$     & $\frac{3^n+3}{2}$                     & $n$ odd $>1$                      & $4$                            & \cite{jiang2022differential}                                             \\
    \hline
    $5$     & $\frac{5^n+3}{2}$                     & any $n$                           & $3$                            & \cite{pang2023differential}                                              \\
    \hline
    $5$     & $\frac{5^n-3}{2}$                     & any $n$                           & $4$ or $5$                     & \cite{yan2021differential}                                               \\
    \hline
    $p$ odd & $\frac{p^n+3}{2}$                     & \makecell[l]{$p\ge 5$,                                                                                                                        \\$p^n\equiv 1\pmod 4$}&$3$ &\cite{sha2022differential}\\
    \hline
    $p$ odd & $\frac{p^n+3}{2}$                     & \makecell[l]{
    $p^n\equiv 3\pmod 4$,                                                                                                                                                                           \\$p\ne 3$}&$2$ or $4$ &\cite{yan2023complete}\\
    \hline
    $p$ odd & $\frac{p^n-3}{2}$                     & \makecell[l]{$p^n\equiv 3\pmod 4$                                                                                                             \\$p^n>7$\\$p^n\ne 27$}             &  $2$ or $3$                                                 & \cite{yan2022class,yan2024class}                    \\
    \hline
    $p$ odd & $p^n-3$                               & any $n$                           & $1\le\delta(f)\le 5$           & \cite{xia2020differential,yan2022differential}
    \\
    \hline
    $p$ odd & $2p^{\frac{n}{2}}-1$                  & $n$ even                          & $p^{\frac{n}{2}}$              & \cite{yan2022note}                                                       \\

    \hline

    $p$ odd & $p^{\frac{n}{2}}+2$                   & \makecell[l]{$n$ even                                                                                                                      }&$2$, $4$ or $p^{\frac{n}{2}}$ &\cite{man2022differential}\\
    \hline
    $p$ odd & $\frac{(p^m+3)}{2}(p^m-1)$            & $n=2m$,\ $\cdots$                 & $p^m-2$                        & \cite{yan2022two}                                                        \\
    \hline
    $p$ odd & $\frac{p^k+1}{2}$                     & $e=\gcd(n,k)$                     & $\frac{p^e-1}{2}$ or $p^e+1$   & \cite{choi2013differential}                                              \\
    \hline
    $p$ odd & $p^{2k}-p^k+1$                        & \makecell[l]{$\gcd(n,k)=e$,                                                                                                                   \\$\frac{n}{e}$\ odd}&$p^e+1$&\cite{yan2019differential,leilei2021}\\
    \hline
    $p$ odd & $\frac{p^n+1}{p^m+1}+\frac{p^n-1}{2}$ & \makecell[l]{$p\equiv 3\pmod 4$,                                                                                                              \\$n$ odd,\\$m\mid n$}&$\frac{p^m+1}{2}$&\cite{choi2013differential}\\
    \hline
    any     & $p^n-2(=-1)$                          & any $n$                           & $\ \ \cdots$                   & \cite{blondeau2010differential,boura2018boomerang,jiang2022differential} \\
    \hline
    any     & $k(p^m-1)$                            & \makecell[l]{$n=2m$,                                                                                                                          \\ $\gcd(k,p^m+1)=1$}                      & $p^m-2$                                      & \cite{hu2023differential}            \\
    \bottomrule
  \end{tabular}
\end{table}

The boomerang attack, introduced by Wagner in \cite{wagner1999boomerang}, is a variation of the differential attack that leverages the differential properties of the upper and lower layers of block ciphers. The resistance of an S-box to boomerang attacks is measured using a concept known as boomerang uniformity. This concept was first introduced by Boura and Canteaut  in \cite{boura2018boomerang} for permutations over binary fields, and was later expanded to general functions over arbitrary finite fields by Li et al. in \cite{li2019new}.

The boomerang uniformity of a function \( f \) over \( \mathbb{F}_q \) is defined as:
\[
\beta_f = \max\limits_{a,b \in \mathbb{F}_q^*}\beta_f(a,b),
\]
where \( \beta_f(a,b) \) represents the number of solutions \( (x,y) \in \mathbb{F}_q^2 \) to the following system of equations:  
\begin{equation}\label{20241219equation1}
  \begin{cases}
  f(x)-f(y)=b,\\
  f(x+a)-f(y+a)=b.
\end{cases}
\end{equation}
Note that the system (\ref{20241219equation1}) is equivalent to the following one:
\begin{equation}
  \begin{cases}
    f(x)-f(y)=b,\\
    D_af(x)=D_a(y).
  \end{cases}
\end{equation}

We can define the boomerang spectrum of a function \( f \) over $\mathbb{F}_q$ as the following multiset:
\[
{\rm{BS}}_f = \{\nu_i:\ 0 \leq i \leq \beta_f\},
\]
where 
\[
\nu_i = \#\left\{(a,b) \in \mathbb{F}_q^* \times \mathbb{F}_q^*:\ \beta_f(a,b) = i\right\}.
\]
Similarly, when examining the boomerang properties of a power function \( f \), it is sufficient to focus on the values \( \beta_f(1,b) \) for \( b \in \mathbb{F}_q^* \). We summarize known results on the boomerang uniformity of power functions in Table \ref{table5}.

\begin{table}[h!]
  \centering
  \caption{Power functions $f(x)=x^d$ over $\mathbb{F}_{p^n}$ with known boomerang uniformity where $p$ is odd\ (quotes in the table indicate omitted content due to length)\\}
  \label{table5}
  \begin{tabular}{llllc}
    \toprule
    $p$     & $d$                                   & Condition                         & $\beta_f$                    & References                                                               \\
    \midrule
    $3$     & $\frac{3^n+3}{2}$                     & $n$ odd                       & $3$                            & \cite{jiang2022differential}                                             \\
    \hline
    $p$ odd & $\frac{p^n-3}{2}$                     & \makecell[l]{$p^n\equiv 3\pmod 4$                                                                                                             }             &  $\le 6$                                              & \cite{yan2022class}                    \\
    \hline
    $p$ odd & $p^{\frac{n}{2}}-1$                     & \makecell[l]{$n$ even,\\$p^{\frac{n}{2}}\not\equiv 2\pmod 3$                                                                                                             }             &  $2$                                              & \cite{yan2022class}                    \\
    \hline
    $p$ odd & $\frac{(p^m+3)}{2}(p^m-1)$            & $n=2m$,\ $\cdots$                 & $2$                        & \cite{yan2022two}                                                        \\
    \hline
    any     & $p^n-2(=-1)$                          & any $n$                           & $2\le\beta_f\le 6$                   & \cite{blondeau2010differential,boura2018boomerang,jiang2022differential} \\
    \hline
    any     & $k(p^m-1)$                            & \makecell[l]{$n=2m$,                                                                                                                          \\ $\gcd(k,p^m+1)=1$}                      & $2$ or $4$                                      & \cite{hu2023differential}            \\
    \hline $p$ odd &$p^{\frac{n}{2}}+2$&\makecell[l]{$n$ even \\
    $p^{\frac{n}{2}}\equiv 1\ \text{or}\ 3\ ({\rm{mod}}\ 6)$        \\ $p\ne 5$                                                                                                             }&$p^{\frac{n}{2}}+2$ or $5$&This paper\\
    \bottomrule
  \end{tabular}
\end{table}
A well-known method of attack in symmetric cryptography is the linear attack. The resistance of an S-box against linear attacks is measured in terms of its nonlinearity, which is closely related to the Walsh spectrum.

For any function \( f \) defined over \( \mathbb{F}_q \), the Walsh transform of \( f \)  at $(a,b)\in\mathbb{F}_q^2$ is given by
\[
W_f(a,b)=\sum_{x\in\mathbb{F}_{q}}\xi_p^{{\rm{Tr}}_{\mathbb{F}_q/\mathbb{F}_p}(bf(x)-ax)}, 
\]
where \( \xi_p=e^{2\pi i/p} \)   is a  primitive complex $p$-th root of unity, $i$ denotes the primitive complex fourth root of unity as usual, and \( {\rm{Tr}}_{\mathbb{F}_q/\mathbb{F}_p} \) is the (absolute) trace function from \( \mathbb{F}_q \) to (its prime field) \( \mathbb{F}_p \). The Walsh spectrum of \( f \) is defined as the multiset
\[
\{W_f(a,b):\ a\in\mathbb{F}_q,\ b\in\mathbb{F}_q^*\},
\]

We have the following well-known properties of the Walsh spectrum:

\begin{lemma}[\cite{carlet2004highly}]\label{20240623lemma1}
For any function \( f:\mathbb{F}_q\rightarrow\mathbb{F}_q \) with \( f(0)=0 \), the following statements hold:
\begin{enumerate}[(1)]
    \item \(\sum_{a\in\mathbb{F}_q,b\in\mathbb{F}_q^*}W_f(a,b)=q^2-q\);
    \item (Parseval's relation) \(\sum_{a\in\mathbb{F}_q}|W_f(a,b)|^2=q^2\) for any \( b\in\mathbb{F}_q \).
  \end{enumerate}
\end{lemma}

A key challenge regarding the Walsh transform is identifying cryptographic functions with only a few distinct values and analyzing their value distributions. Numerous studies have been conducted on this topic, some of which are listed in Table \ref{table2}.

\begin{table}[h!]
  \centering
  \caption{Some power functions $f(x)=x^d$ over $\mathbb{F}_{p^n}$ whose Walsh spectrum takes only a few distinct values\\}
  \label{table2}
  \begin{tabular}{llcc}
    \toprule
    $d$                                      & Conditions                                                                       & Valued & References           \\
    \hline
    $(\frac{p^{\frac{n}{4}}+1}{2})^2$        & $4\mid n$                                                                        & $4$    & \cite{seo2008cross}  \\
    \hline
    $(\frac{p^{\frac{n}{2}}+1}{2})^2$        & $2\mid n$                                                                        & $4$    & \cite{luo2010cross}  \\
    \hline
    $\frac{p^n+1}{p+1}+\frac{p^n-1}{2}$      & \makecell[l]{$p\equiv 3\pmod{4}$,                                                                                \\$n$ odd}&$9$&\cite{xia2010further}\\
    \hline
    $\frac{p^n+1}{p^k+1}+\frac{p^n-1}{2}$    & \makecell[l]{$p\equiv 3\pmod{4}$,                                                                                \\$n$ odd,\\$k\mid n$}&$9$&\cite{choi2013cross}\\
    \hline
    $\frac{p^k+1}{2}$                        & $\frac{k}{\gcd(n,k)}$ odd                                                        & $9$    & \cite{luo2008cyclic} \\
    \hline
    $\frac{(p^{\frac{n}{2}+1})^2}{2(p^k+1)}$ & \makecell[l]{$p\equiv 3\pmod{4}$,                                                                                \\$n\equiv 2\pmod{4}$,\\$k\mid\frac{n}{2}$}&$6$&\cite{luo2011two}\\
    \hline
    $\frac{2^{\frac{n}{2}}+1}{3}$            & $n\equiv2\pmod{4}$                                                               & $3$    & \cite{ness2006cross} \\
    \hline
    $\frac{2^{\frac{nl}{2}}+1}{2^l+1}$       & \makecell[l]{$n\equiv 2\pmod{4}$,                                                                                \\$l$ odd,\\$\gcd(n,l)=1$}&$3$&\cite{luo2016binary}\\
    \hline
    \makecell[l]{$d(p^k+1)\equiv 2\pmod{p^n-1}$                                                                                                                 \\$d\equiv 1+\frac{p^e-1}{2}\pmod{p^e-1}$}&\makecell[l]{$\frac{n}{e}>3$ odd\\$p^e\equiv 3\pmod{4}$,\\where $e=\gcd(n,k)$}&$9$&\cite{xia2014some}\\
    \hline
    $2p^{\frac{n}{2}}-1$,                    & \makecell[l]{                                  $p^{\frac{n}{2}}\equiv 2\pmod{3}$                                 \\$n$ even} & $4$    &  \cite{li2022class}          \\
    \hline
    $3^{\frac{n}{2}}+2$,               & $p=3$, $n$ even                                                                         & $4$    & This paper           \\
    \bottomrule
  \end{tabular}
\end{table}

In this paper, we focus on the power function \( f(x) = x^{q+2} \) defined over the finite field \( \mathbb{F}_{q^2} \), where \( q \) is an odd prime power. The differential spectrum of this function has been determined in \cite{man2022differential}.  Section \ref{section2} presents an alternative method to determine the differential spectrum of \( f \), which involves simpler computations. Section \ref{section3} employs a similar strategy to analyze the boomerang uniformity of \( f \) for values of \( q \) that are congruent to \( 1 \) or \( 3 \mod 6 \), specifically excluding the case where \( p = 5 \) and \( m \) is even.
In Section \ref{section4}, we examine the value distribution of the Walsh spectrum of \( f \) for \( p = 3 \) and utilize the obtained results  to determine the weight distribution of a \( 4 \)-weight cyclic code. Before these analyses, Section \ref{section1} introduces some preliminary definitions and results. Finally, Section \ref{section5} conclusions the paper.

\section{Preliminaries}\label{section1}

Let \( q = p^m \), where \( p \) is an odd prime and \( m \) is a positive integer. We will denote by  \( \mathbb{U} \) the subset  \( \{ z \in \mathbb{F}_{q^2}:\ z^{q + 1} = 1 \} \) within \( \mathbb{F}_{q^2} \) for the following discussion. The lemma presented below will play a crucial role in Section \ref{section4} and is often very useful for studying problems over \( \mathbb{F}_{q^2} \).

\begin{lemma}[{\cite[Lemma 2]{yan2022note}}]\label{20240617lemma1}
  For any square element $x\in\mathbb{F}_{q^2}^*$, there exist exactly two pairs, namely $(y,z)$ and $(-y,-z)$ such that $x=yz=(-y)(-z)$, $\pm y\in\mathbb{F}_{q}^*$ and $\pm z\in\mathbb{U}$.
\end{lemma}

The following lemma is a simple consequence of the law of quadratic reciprocity.

\begin{lemma}[{\cite[Lemma 3]{man2022differential}}]\label{20240620lemma6}
  Assume that $p>3$. Then $-3$ is a non-square element in $\mathbb{F}_q$ if and only if $q\equiv 5\pmod{6}$.
\end{lemma}

Let $C_0$ ($C_1$, resp.) be the subset of $\mathbb{F}_q^*$ consisting of square (non-square, resp.) elements. Let $\eta:\mathbb{F}_q\rightarrow\mathbb{C}$ be the quadratic character of $\mathbb{F}_q$ given by 
$$\eta(x)=\left\{\begin{array}{ll}
  0, & \text{if $x=0$}, \\
    1, & \text{if $x\in C_0$}, \\
    -1, & \text{if $x\in C_1$}.
  \end{array}\right.$$
For any $i,j\in\{0,1\}$, we put $(i,j)=\#\Big((C_i+1)\cap C_j\Big)$. We have the following result on the values of $(i,j)$.

\begin{theorem}[{\cite[Lemma 1.5]{ding2014codes}}]\label{20241218them11}
  If $q\equiv 1\ ({\rm{mod}}\ 4)$, then 
  $$(0,0)=\frac{q-5}{4},\ (0,1)=(1,0)=(1,1)=\frac{q-1}{4}.$$
  If $q\equiv 3\ ({\rm{mod}}\ 4)$, then
  $$(0,1)=\frac{q+1}{r},\ (0,0)=(1,0)=(1,1)=\frac{q-3}{4}.$$
\end{theorem}

Next, we present results on polynomial factorization over finite fields. The following is a helpful result regarding the number of monic irreducible factors of a polynomial over finite fields.

\begin{theorem}[Stickelberger's parity theorem,\ {\cite{stickelberger1898uber}}]\label{20241219them12}
  Let $f\in\mathbb{F}_q[x]$ be a polynomial with degree $\deg(f)=n\ge 2$ and discriminant $D(f)\ne 0$, and let $k$ be the number of monic irreducible factors of $f$. Then $\eta\big(D(f)\big)=(-1)^{n-k}$.
\end{theorem}

\begin{remark}
  Assume that $f(x)=a_0(x-\alpha_1)\cdots(x-\alpha_n)$ with $a_0\in\mathbb{F}_q$ and $\alpha_1,\cdots,\alpha_n$ in the splitting field of $f$ over $\mathbb{F}_q$. Then the discriminant $D(f)$ of $f$ is defined as 
  $$D(f)=a_0^{2n-2}\displaystyle\prod_{1\le i<j\le n}(\alpha_i-\alpha_j)^2.$$
  It is clear that $D(f)=0$ if and only if $f$ has a multiple root in its splitting field over $\mathbb{F}_q$. 
  
  Moreover, it is well-known that $D(f)\in\mathbb{F}_q$ can be expressed in terms of the coefficients of $f$. For a more detailed introduction to the discriminant, readers may refer to \cite{lidl1997finite}.
\end{remark}

\begin{definition}
  Let $f(x)\in\mathbb{F}_q[x]$. If $f(x)=f_1(x)f_2(x)\cdots f_k(x)$ is the factorization of $f$ into irreducible factors over $\mathbb{F}_q$ with $i_1=\deg(f_1)\le i_2=\deg(f_2)\le\cdots\le i_k=\deg(f_k)$, then we say that $f$ is of type $(i_1,i_2,\cdots,i_k)$.
\end{definition}

We have a complete theory for the factorization of cubic polynomials.

\begin{theorem}[{\cite{dickson1906criteria}}]\label{20241219them10}
  Assume that $p>3$, let $\omega$ be a square root of $-3$ in $\mathbb{F}_{q^2}$, and let $a,b\in\mathbb{F}_q$ be such that $-4a^3-27b^2\ne 0$. Then the factorization of the cubic polynomial $f(x)=x^3+ax+b$ over $\mathbb{F}_q$ is characterized as follows:
  \begin{enumerate}[(1)]
    \item $f$ is of type $(1,1,1)$ if and only if $-4a^3-27b^2=81c^2$ for some $c\in\mathbb{F}_q$ and $\frac{1}{2}(-b+c\omega)$ is a cube in $\mathbb{F}_q$ $(\mathbb{F}_{q^2}$, resp.$)$ when $q\equiv 1\ ({\rm{mod}}\ 3)$ $\big(q\equiv 2\ ({\rm{mod}}\ 3)$, resp.$\big)$;
    \item $f$ is of type $(1,2)$ if and only if $-4a^3-27b^2\in C_1$;
    \item $f$ is of type $(3)$ if and only if $-4a^3-27b^2=81c^2$ for some $c\in\mathbb{F}_q$ and $\frac{1}{2}(-b+c\omega)$ is not a cube in $\mathbb{F}_q$ $(\mathbb{F}_{q^2}$, resp.$)$ when $q\equiv 1\ ({\rm{mod}}\ 3)$ $\big(q\equiv 2\ ({\rm{mod}}\ 3)$, resp.$\big)$.
  \end{enumerate}
\end{theorem}

\begin{remark}
  Note that $-4a^3-27b^2$ is the discriminant of the cubic polynomial $f(x)=x^3+ax+b$. Hence, if it is zero, then $f$ has a multiple root in its splitting field over $\mathbb{F}_q$.
\end{remark}

\begin{theorem}[{\cite[Theorem 2]{williams1975note}}]\label{20241219them7}
  Let $a,b\in\mathbb{F}_{3^m}$ with $a\ne 0$. The factorization of the cubic polynomial $f(x)=x^3+ax+b$ over $\mathbb{F}_{3^m}$ is characterized as follows:
  \begin{enumerate}[(1)]
    \item $f$ is of type $(1,1,1)$ if and only if $-a$ is a square element in $\mathbb{F}_{3^m}$ and ${\rm{Tr}}_{\mathbb{F}_q/\mathbb{F}_3}(b/c^3)=0$, where $c$ is a square root of $-a$;
    \item $f$ is of type $(1,2)$ if and only if $-a$ is a non-square element in $\mathbb{F}_{3^m}$;
    \item $f$ is of type $(3)$ if and only if $-a$ is a square element in $\mathbb{F}_{3^m}$ and ${\rm{Tr}}_{\mathbb{F}_q/\mathbb{F}_3}(b/c^3)$\\
    $\ne 0$, where $c$ is a square root of $-a$.
  \end{enumerate}
\end{theorem}

In the later part of the paper, we will frequently use character sums, which help calculate or estimate the number of elements that satisfy certain conditions. 

First is the most fundamental type of character sum: the Gaussian sum.

\begin{definition}[Gaussian sum]
  Let $\psi$ be an additive and $\chi$ a multiplicative character of $\mathbb{F}_q$. Then the Gaussian sum $G(\chi,\psi)$ is defined as
  $$G(\chi,\psi)=\sum\limits_{c\in\mathbb{F}_q}\chi(c)\psi(c).$$
\end{definition}
\begin{remark}
  A well-known result is that if $\chi$ and $\psi$ are both non-trivial, then $|G(\chi,\psi)|=\sqrt{q}$ (see \cite[Theorem 5.11]{lidl1997finite}).
\end{remark}

Gaussian sums are often used to compute other types of character sums.

\begin{proposition}[{\cite[Theorem 5.33]{lidl1997finite}}]\label{20241218prop1}
  Let $\psi$ be a non-trivial additive character of $\mathbb{F}_q$ and let $f(x)=a_2x^2+a_1x+a_0\in\mathbb{F}_q[x]$ with $a_2\ne 0$. Then 
  $$\sum\limits_{c\in\mathbb{F}_q}\psi\big(f(c)\big)=\psi(a_0-\frac{a_1^2}{4a_2})\eta(a_2)G(\eta,\psi),$$
  where $\eta$ is the quadratic character of $\mathbb{F}_q$.
\end{proposition}

The following is a multiplicative version of Proposition \ref{20241218prop1}.

\begin{proposition}[{\cite[Theorem 5.48]{lidl1997finite}}]\label{20241219prop5}
  Let $f(x)=a_2x^2+a_1x+a_0\in\mathbb{F}_q[x]$ with $a_2\ne 0$. Put $d=a_1^2-4a_0a_2$. Then 
  $$\sum\limits_{c\in\mathbb{F}_q}\eta\big(f(c)\big)=\begin{cases}
  -\eta(a_2),&\text{if }d\ne 0,\\
  (q-1)\eta(a_2),&\text{if }d=0.
\end{cases}$$
\end{proposition}

The following Weil bound is a powerful tool for estimating character sums.

\begin{theorem}[{\cite[Theorem 5.41]{lidl1997finite}}]\label{20240904them3}
  Let $q$ be an odd prime power, let $\chi$ be a multiplicative character of $\mathbb{F}_q$ of order $n>1$, let $f\in\mathbb{F}_q[x]$ be a monic polynomial of positive degree that is not an $n$-th power of a polynomial, and let $d$ be the number of distinct roots of $f$ in its splitting field over $\mathbb{F}_q$. Then for any $a\in\mathbb{F}_q$, we have
  $$\left|\sum\limits_{x\in\mathbb{F}_q}\chi\big(af(x)\big)\right|\le(d-1)\sqrt{q}.$$
\end{theorem}

Next, we shall use character sums to evaluate the size of certain sets. Before that, we introduce a simple lemma, which is a special case of Hilbert's Theorem 90.

\begin{lemma}[{\cite[Theorem 2.25]{lidl1997finite}}]\label{20241218lemma1}
  Let $F$ be a finite extension of $K=\mathbb{F}_q$. Then for any $\alpha\in F$, ${\rm{Tr}}_{F/K}(\alpha)=0$ if and only if $\alpha=\beta^q-\beta$ for some $\beta\in F$.
\end{lemma}

\begin{lemma}\label{20241219lemma5}
  Assume that $p=3$. Then we have 
 $$\#\left\{x\in\mathbb{F}_q^*:\ x,x+1,x-1\in C_0\right\}\le\frac{q+2\sqrt{q}+9}{8}.$$
\end{lemma}
\begin{proof}
  We have 
  \begin{align*}
    &\#\left\{x\in\mathbb{F}_q^*:\ x,x+1,x-1\in C_0\right\}\\
    =\ &\frac{1}{8}\sum\limits_{x\in\mathbb{F}_q^*\setminus\{\pm 1\}}\big(1+\eta(x)\big)\big(1+\eta(x+1)\big)\big(1+\eta(x-1)\big)\\
    \le\ &\frac{1}{8}\Big(12+\sum\limits_{x\in\mathbb{F}_q}\big(1+\eta(x)\big)\big(1+\eta(x+1)\big)\big(1+\eta(x-1)\big)\Big)\\
    =\ &\frac{q+12}{8}+\frac{1}{8}\Big(\sum\limits_{x\in\mathbb{F}_q}\eta(x^2-x)+\eta(x^2+x)+\eta(x^2-1)+\sum\limits_{x\in\mathbb{F}_q}\eta(x^3-x)\Big)\\
    =\ &\frac{q+9}{8}+\frac{1}{8}\cdot\sum\limits_{x\in\mathbb{F}_q}\eta(x^3-x)\quad(\mbox{by Proposition \ref{20241219prop5}})\\
    \le\ &\frac{q+9}{8}+\frac{2\sqrt{q}}{8}=\frac{q+2\sqrt{q}+9}{8}\quad(\mbox{by Theorem \ref{20240904them3}}).
  \end{align*}
\end{proof}

\begin{lemma}\label{20241219lemma6}
  Assume that $p=3$. Then for any $i\in\{0,1,2\}$, we have
  \begin{align*}
    \#\left\{x\in\mathbb{F}_q:\ {\rm{Tr}}_{\mathbb{F}_q/\mathbb{F}_3}(x)=i,\ {\rm{Tr}}_{\mathbb{F}_q/\mathbb{F}_3}(x^2)=0\right\}\ge\frac{q-6\sqrt{q}}{9}.
  \end{align*}
\end{lemma}

\begin{proof}
  Let $\psi_1$ be the canonical additive character of $\mathbb{F}_q$ defined by $\psi_1(x)=\xi_3^{{\rm{Tr}}_{\mathbb{F}_q/\mathbb{F}_3}(x)}$, and let $x_i\in\mathbb{F}_q$ be such that ${\rm{Tr}}_{\mathbb{F}_q/\mathbb{F}_3}(x_i)=i$ for any $i\in\{0,1,2\}$. We have
  \begin{align*}
    N_1:=\ &\#\left\{x\in\mathbb{F}_q:\ {\rm{Tr}}_{\mathbb{F}_q/\mathbb{F}_3}(x)=0,\ {\rm{Tr}}_{\mathbb{F}_q/\mathbb{F}_3}(x^2)=0\right\}\\
    =\ &\frac{1}{9}\sum\limits_{x\in\mathbb{F}_q}\Big(\sum\limits_{a\in\mathbb{F}_3}\xi_3^{a{\rm{Tr}}_{\mathbb{F}_q/\mathbb{F}_3}(x^2)}\Big)\Big(\sum\limits_{b\in\mathbb{F}_3}\xi_3^{b\big({\rm{Tr}}_{\mathbb{F}_q/\mathbb{F}_3}(x)-i\big)}\Big)\\
    =\ &\frac{1}{9}\sum\limits_{a,b\in\mathbb{F}_3}\sum\limits_{x\in\mathbb{F}_q}\xi_3^{{\rm{Tr}}_{\mathbb{F}_q/\mathbb{F}_3}(ax^2+bx-bx_i)}\\
    =\ &\frac{1}{9}\sum\limits_{a\in\mathbb{F}_3^*}\sum\limits_{b\in\mathbb{F}_3}\sum\limits_{x\in\mathbb{F}_q}\xi_3^{{\rm{Tr}}_{\mathbb{F}_q/\mathbb{F}_3}(ax^2+bx-bx_i)}+\frac{1}{9}\sum\limits_{b\in\mathbb{F}_3}\sum\limits_{x\in\mathbb{F}_q}\xi_3^{{\rm{Tr}}_{\mathbb{F}_q/\mathbb{F}_3}(bx-bx_i)}\\
    =\ &\frac{q}{9}+\frac{1}{9}\sum\limits_{a\in\mathbb{F}_3^*}\sum\limits_{b\in\mathbb{F}_3}\sum\limits_{x\in\mathbb{F}_q}\psi_1(ax^2+bx-bx_i)\\
    =\ &\frac{q}{9}+\frac{1}{9}\sum\limits_{a\in\mathbb{F}_3^*}\sum\limits_{b\in\mathbb{F}_3}\psi_1(-bx_i-\frac{b^2}{a})\eta(a)G(\eta,\psi_1)\quad(\mbox{by Proposition \ref{20241218prop1}})\\
    =\ &\frac{q}{9}+\frac{G(\eta,\psi_1)}{9}\sum\limits_{a\in\mathbb{F}_3^*}\eta(a)\sum\limits_{b\in\mathbb{F}_3}\psi_1(-bx_i-\frac{b^2}{a}),
  \end{align*}
  which implies that 
  $$N_1\ge\frac{q}{9}-\frac{2}{3}\cdot|G(\eta,\psi_1)|=\frac{q}{9}-\frac{2}{3}\sqrt{q}.$$
\end{proof}

The following lemma will help prove certain properties of the boomerang spectrum of the power function $f(x)=x^{q+2}$ when $p=3$.

\begin{lemma}\label{20241219lemma7}
  Assume that $p=3$ and $m\ge 4$. For any $i\in\{0,1,2\}$, there exists $x\in C_1$ such that $x^3-x=b^2$ for some $b\in\mathbb{F}_q^*$ with ${\rm{Tr}}_{\mathbb{F}_q/\mathbb{F}_3}(b)=i$.
\end{lemma}
\begin{proof}

Let \( E(q) = \{(x,y) \in \mathbb{F}_q^* \times \mathbb{F}_q^*:\ x^3 - x = y^2\} \). It is important to note that \( (x \pm 1)^3 - (x \pm 1) = x^3 - x \). 

If there exists \( x \in \mathbb{F}_q^* \) such that \( x + \epsilon \in C_1 \) and \( x^3 - x = b^2 \) for some \( b \in \mathbb{F}_q^* \) with \( \text{Tr}_{\mathbb{F}_q/\mathbb{F}_3}(b) = i \), where \( \epsilon \in \{ \pm 1 \} \), then \( x' = x + \epsilon \) is the desired element. Hence, if this lemma does not hold for \( i \in \{0, 1, 2\} \), we must have
\[
E_1 := \{(x,y) \in E(q):\ \text{Tr}_{\mathbb{F}_q/\mathbb{F}_3}(y) = i\} 
\subset E_2 := \{(x,y) \in E(q):\ x, x \pm 1 \in C_0\}.
\]

It is clear that
\[
\# E_2 = 2 \cdot \# \{x \in \mathbb{F}_q^*:\ x, x \pm 1 \in C_0\},
\]
and, by Lemma \ref{20241218lemma1},
\[
\# E_1 = 3 \cdot \# \left\{y \in \mathbb{F}_q^*:\ \text{Tr}_{\mathbb{F}_q/\mathbb{F}_3}(y) = i, \ \text{Tr}_{\mathbb{F}_q/\mathbb{F}_3}(y^2) = 0\right\}.
\]
By applying Lemmas \ref{20241219lemma5} and \ref{20241219lemma6}, we arrive at the following inequality:
\[
3\left(\frac{q - 6\sqrt{q}}{9} - 1\right) \leq 2\left(\frac{q + 2\sqrt{q} + 9}{8}\right),
\]
which simplifies to
\[
q - 30\sqrt{q} - 63 \leq 0.
\]

This inequality does not hold for \( q \geq 1023 \). Therefore, this lemma is valid for \( q \geq 1023 \). We verify this lemma directly using computational methods for \( 81 \leq q < 1023 \).
\end{proof}

\section{The differential spectrum of $x^{q+2}$ over $\mathbb{F}_{q^2}$}\label{section2}

In this section, we focus on the differential spectrum of the power function \( f(x) = x^{q+2} \) for \( x \in \mathbb{F}_{q^2} \), where \( q = p^m \), \( p \) is an odd prime, and \( m \) is a positive integer. This differential spectrum of $f$ was determined by Man, Xia, Li, and Helleseth in 2022 \cite{man2022differential}. This section will present an alternative method for evaluating the differential spectrum of $f$, which involves simplified calculations.

Note that
\begin{align}
  D_1f(x) & =  f(x+1)-f(x)   \notag                                                                             \\
          & = 2x^{q+1}+x^q+x^2+2x+1\notag                                                                   \\
          & =  2(x^{q+1}+\frac{1}{2}x^q+\frac{1}{2}x+\frac{1}{4})+(x^2+x+\frac{1}{4})+\frac{1}{4} \notag                       \\
          & =  2(x+\frac{1}{2})^{q+1}+(x+\frac{1}{2})^{2}+\frac{1}{4}\notag                                    \\
          & =  2u^{q+1}+u^{2}+\frac{1}{4},\label{20240628equation1}
\end{align}
where $u=x+\frac{1}{2}$. For any $b\in\mathbb{F}_{q^2}$, we set 
\begin{align}
  \delta(b)  :=\delta_f(1,b+\frac{1}{4}) & =\#\{x\in\mathbb{F}_{q^2}:\ D_1f(x)=b+\frac{1}{4}\}\notag                 \\
                                         & =\#\{u\in\mathbb{F}_{q^2}:\ 2u^{q+1}+u^{2}=b\}.\label{20240620equation1}
\end{align}

Let \(\alpha\) be a fixed non-square element in \(\mathbb{F}_{q}\), and let \(Z\in\mathbb{F}_{q^2} \setminus \mathbb{F}_q\) be such that \(Z^2 = \alpha\). Then any element in \(\mathbb{F}_{q^2}\) can be uniquely expressed in the form \(c + dZ\) where \(c, d \in \mathbb{F}_q\). By Lemma \ref{20240620lemma6}, if \(q \equiv 5 \,({\rm{mod}}\, 6)\), we can take \(\alpha\) to be \(-3\). 

Since \(Z^{q-1} = \alpha^{\frac{q-1}{2}} = -1\), it follows that \(Z^q = -Z\), which implies the following  equations:
\begin{equation}\label{20241018equation3}
  \traced(c+dZ)=(c+dZ)+(c-dZ)=2c,
\end{equation}
and 
\begin{equation}
  \norm(c+dZ)=(c+dZ)(c-dZ)=c^2-d^2\alpha.
\end{equation}
Here, \( \traced \) and \( \norm \) denote the trace and norm functions from \(\mathbb{F}_{q^2}\) to \(\mathbb{F}_q\), respectively.

 Moreover, we have 
\begin{equation}\label{20241018equation1}
  2(x+yZ)^{q+1}+(x+yZ)^{2}=(3x^2-y^2\alpha)+2xyZ
\end{equation}
for any $x,y\in\mathbb{F}_q$. By equations (\ref{20240620equation1}) and (\ref{20241018equation1}), we obtain that 
\begin{equation}\label{20241018equation2}
  \delta(c+dZ)=\#\left\{(x,y)\in\mathbb{F}_q^2:\ \begin{cases}
    3x^2-y^2\alpha=c\\
    2xy=d
  \end{cases}\right\}
\end{equation}
for any $c,d\in\mathbb{F}_q$. If $c\in\mathbb{F}_q^*$, then
\begin{align}
  \delta(c)&=\#\left\{(x,y)\in\mathbb{F}_q^2:\ \begin{cases}
    3x^2-y^2\alpha=c\\
    xy=0
  \end{cases}\right\}\notag\\
&=\#\left\{x\in\mathbb{F}_q:\ 3x^2=c\right\}+\#\left\{y\in\mathbb{F}_q:\ y^2=-\frac{c}{\alpha}\right\}.\label{20241018equation10}
\end{align}
For any $c\in\mathbb{F}_q$ and $d\in\mathbb{F}_q^*$, we have 
\begin{align}
  \delta(c+dZ)&=\#\left\{x\in\mathbb{F}_q^*:\ 3x^2-\frac{d^2\alpha}{4x^2}=c\right\}\notag\\
&=\#\left\{x\in\mathbb{F}_q^*:\ 3x^4-cx^2-\frac{d^2\alpha}{4}=0\right\}\notag\\
&=2\cdot\#\left\{y\in C_0:\ 3y^2-cy-\frac{d^2\alpha}{4}=0\right\}.\label{20241018equation11}
\end{align}

In particular, we have the following conclusion.

\begin{lemma}\label{20241018lemma4}
  For any $b\in\mathbb{F}_{q^2}^*$, $\delta(b)$ is an even number such that $\delta(b)\le 4$.
\end{lemma}

\begin{proposition}\label{20240620prop4}
  We have
  $$\delta(0)=\begin{cases}
      q, & \text{if }p=3,    \\
      1, & \text{otherwise}.
    \end{cases}$$
\end{proposition}

\begin{proof}
  It is clear that
  \begin{align*}
    \delta(0)&=\#\left\{(x,y)\in\mathbb{F}_q^2:\ \begin{cases}
      3x^2-y^2\alpha=0\\
      xy=0
    \end{cases}\right\}\\
    &=\begin{cases}
      \#(\mathbb{F}_q\times\{0\})=q, & \text{if }p=3,    \\
      \#\{(0,0)\}=1, & \text{otherwise}.
    \end{cases}
  \end{align*}
\end{proof}

The following proposition completely describes the values $\delta(b)$ ($b\in\mathbb{F}_{q^2}^*$) when $p=3$.

\begin{proposition}\label{20240626prop1}
  If $p=3$, then for any $b\in\mathbb{F}_{q^2}^*$, we have
  \begin{equation}\label{20241018equation4}
    \delta(b)=\begin{cases}
      2, & \text{if }\traced(b)\ \text{is a non-square element in}\ \mathbb{F}_{q}, \\
      0, & \text{if }\traced(b)\ \text{is a square element in}\ \mathbb{F}_{q}.     \\
    \end{cases}
  \end{equation}
   Moreover, there are $\frac{q^2+q-2}{2}$ elements $b\in\mathbb{F}_{q^2}^*$ for which $\delta(b)=0$, and $\frac{q^2-q}{2}$ elements $b\in\mathbb{F}_{q^2}^*$ for which $\delta(b)=2$.
  \end{proposition}

\begin{proof}
  By equation (\ref{20241018equation2}), for any $c,d\in\mathbb{F}_q$, we have 
  \begin{equation*}
    \delta(c+dZ)=\#\left\{(x,y)\in\mathbb{F}_q^2:\ \begin{cases}
      y^2=\frac{2c}{\alpha}\\
      xy=2d
    \end{cases}\right\}.
  \end{equation*}
  Then it is clear that for any $d\in\mathbb{F}_q^*$, we have $\delta(dZ)=0$. Moreover, for any $c\in\mathbb{F}_q^*$ and $d\in\mathbb{F}_q$, we have
  \begin{align*}
    &\left\{(x,y)\in\mathbb{F}_q^2:\ \begin{cases}
      y^2=\frac{2c}{\alpha}\\
      xy=2d
    \end{cases}\right\}\\
    =\ &\begin{cases}
      \emptyset,&\mbox{if }2c\ \mbox{is square in }\mathbb{F}_q,\\
      \left\{(\frac{2d}{\sqrt{\frac{2c}{\alpha}}},\sqrt{\frac{2c}{\alpha}}),(\frac{2d}{-\sqrt{\frac{2c}{\alpha}}},-\sqrt{\frac{2c}{\alpha}})\right\},&\mbox{otherwise},
    \end{cases}
  \end{align*} 
  where $\pm\sqrt{\frac{2c}{\alpha}}$ are the two square roots of $\frac{2c}{\alpha}$. Then equation (\ref{20241018equation4}) follows immediately from equation (\ref{20241018equation3}).

Next, we want to compute how many elements \( b \in \mathbb{F}_{q^2}^* \) satisfy \( \delta(b) = 2 \), which is equivalent to saying that \( \traced(b) \) is a non-square element in \( \mathbb{F}_q \). Recall that the function \( \traced: \mathbb{F}_{q^2} \to \mathbb{F}_q \) is a surjective \( \mathbb{F}_q \)-linear map, and its kernel contains \( q \) elements. Consequently, each element in \( \mathbb{F}_q \) has \( q \) preimages under $\traced$. Since there are \( \frac{q-1}{2} \) non-square elements in \( \mathbb{F}_q \), it follows that there are \( \frac{q(q-1)}{2} = \frac{q^2 - q}{2} \) elements \( b \in \mathbb{F}_{q^2}^* \) for which \( \delta(b) = 2 \).

As a result, the number of elements \( b \in \mathbb{F}_{q^2}^* \) such that \( \delta(b) = 0 \) can be computed as follows: 
\[
q^2 - 1 - \frac{q^2 - q}{2} = \frac{q^2 + q - 2}{2}.
\] 
Thus, there are \( \frac{q^2 + q - 2}{2} \) elements \( b \in \mathbb{F}_{q^2}^* \) with \( \delta(b) = 0 \).
\end{proof}

\begin{remark}
  We can also use equation (\ref{20240622equation1}) to compute the number of elements $b \in \mathbb{F}_{q^2}^*$ such that $\delta(b) = 0$ and $\delta(b) = 2$, respectively, once we establish that $\delta(b) \in \{0, 2\}$ for any $b \in \mathbb{F}_{q^2}^*$.
\end{remark}

Next, we address the case where $p>3$.

\begin{proposition}
  Assume that $p>3$. Then
  \begin{equation}\label{20241018equation9}
    \delta(3)=\begin{cases}
      4, & \text{if }q\equiv 5\ ({\rm{mod}}\ 6), \\
      2, & \text{otherwise}.
    \end{cases}
  \end{equation}
  In particular, if $q\equiv 5\pmod{6}$, then $\delta_f=4$.
\end{proposition}
\begin{proof}
  By equation (\ref{20241018equation10}), we have
  \begin{align*}
    \delta(3)&=\#\left\{x\in\mathbb{F}_q:\ x^2=1\right\}+\#\left\{y\in\mathbb{F}_q:\ y^2=-\frac{3}{\alpha}\right\}\\
    &=\begin{cases}
      4,&\mbox{if }-3\ \mbox{is a non-square element in }\mathbb{F}_q,\\
      2,&\mbox{otherwise}.
    \end{cases}
  \end{align*}
  Then equation (\ref{20241018equation9}) follows from Lemma \ref{20240620lemma6},  and the second assertion follows from Lemma \ref{20241018lemma4}.
\end{proof}

\begin{proposition}\label{20240626prop2}
  Assume that $q\equiv 5\pmod{6}$. Then there are exactly $q-1$ elements $b\in\mathbb{F}_{q^2}^*$ such that $\delta(b)=2$.
\end{proposition}

\begin{proof}
  Since $q\equiv 5\pmod{6}$, $-3$ is a non-square element in $\mathbb{F}_q$. It follows that for any
  $c\in\mathbb{F}_q^*$, either both of $\frac{c}{3}$ and $\frac{-c}{\alpha}$ are in $C_0$ or neither of them is in $C_0$. By equation (\ref{20241018equation10}), we have $\delta(c)=0$ or $4$. For any $c\in\mathbb{F}_q$ and $d\in\mathbb{F}_q^*$, put $f_{c,d}(y)=3y^2-cy-\frac{d^2\alpha}{4}$. Then by equation (\ref{20241018equation11}), $\delta(c+dZ)=2$ if and only if one of the following two cases occurs:
  \begin{enumerate}[(1)]
    \item $f_{c,d}(y)$ has exactly one root in $\mathbb{F}_q^*$ and it is in $C_0$;
    \item $f_{c,d}(y)$ has two roots in $\mathbb{F}_q^*$ and exactly one of them is in $C_0$.
  \end{enumerate} 
  Let $y_1,y_2$ be the two roots of $f_{c,d}(y)$ in $\mathbb{F}_{q^2}^*$. Then $y_1y_2=\frac{-d^2\alpha}{12}\in C_0$. Hence, case (2) cannot occur. Note that $f_{c,d}(y)$ has exactly one root in $\mathbb{F}_q^*$ if and only if the discriminant $\Delta=c^2+3d^2\alpha=0$. Moreover, in this case, the only root of $f_{c,d}(y)$ is $\frac{c}{6}$. Hence $\delta(c+dZ)=2$ if and only if $\frac{c}{6}\in C_0$ and $d^2=-\frac{c^2}{3\alpha}$. There are $2\cdot\# C_0=q-1$ such elements.   This completes the proof.
\end{proof}

\begin{remark}
  In this case, we can take $\alpha=-3$. Then it follows from the proof that for any $b\in\mathbb{F}_{q^2}^*$, $\delta(b)=2$ if and only if $b=2c(3+\omega)$ or $2c(3-\omega)$ for some $c\in C_0$, where $\omega$ is a square root of $-3$ in $\mathbb{F}_{q^2}$. The latter condition is equivalent to saying that one of the following two elements is a square element in $\mathbb{F}_q$:
  $$\frac{b}{2(3+\omega)}\quad\mbox{and}\quad\frac{b}{2(3-\omega)}.$$
\end{remark}

Summarizing the previous results, we obtain the following main theorem of this section.
\begin{theorem}\label{20240629them1}
  \begin{enumerate}[(1)]
    \item If $p=3$, then $\delta_f=q$ and the differential spectrum of $f$ is
          $${\rm{DS}}_f=\{\omega_0=\frac{q^2+q-2}{2},\ \omega_2=\frac{q^2-q}{2},\ \omega_q=1\}.$$
          In particular, $f$ is locally APN.
    \item If $q\equiv 1\pmod{6}$, then $\delta_f=2$ and the differential spectrum of $f$ is
          $${\rm{DS}}_f=\{\omega_0=\frac{q^2-1}{2},\ \omega_1=1,\ \omega_2=\frac{q^2-1}{2}\}.$$
          In particular, $f$ is APN.
    \item If $q\equiv 5\pmod{6}$, then $\delta_f=4$ and the differential spectrum of $f$ is
          $${\rm{DS}}_f=\{\omega_0=\frac{(3q+1)(q-1)}{4},\ \omega_1=1,\ \omega_2=q-1,\ \omega_4=\frac{(q-1)^2}{4}\}.$$
  \end{enumerate}
\end{theorem}

\begin{proof}

Statement (1) follows directly from Proposition \ref{20240620prop4} and Proposition \ref{20240626prop1}. We will now prove statement (2). By using equation (\ref{20240622equation1}), we have the system of equations:
$$
\begin{cases}
    \omega_0 + \omega_1 + \omega_2 = q^2, \\
    \omega_1 + 2\omega_2 = q^2.
\end{cases}
$$
From Proposition \ref{20240620prop4} and Lemma \ref{20241018lemma4}, we find that $\omega_1 = 1$. Consequently, we can easily obtain that $\omega_0 = \frac{q^2 - 1}{2}$ and $\omega_2 = \frac{q^2 - 1}{2}$.

Finally, we prove Statement (3). Again, using equation (\ref{20240622equation1}), we have the following system of equations:
$$
\begin{cases}
    \omega_0 + \omega_1 + \omega_2 + \omega_4 = q^2, \\
    \omega_1 + 2\omega_2 + 4\omega_4 = q^2.
\end{cases}
$$
From Proposition \ref{20240620prop4}, Lemma \ref{20241018lemma4} and Proposition \ref{20240626prop2}, we find that $\omega_1 = 1$ and $\omega_2 = q - 1$. Then it is straightforward to derive that $\omega_0 = \frac{(3q + 1)(q - 1)}{4}$ and $\omega_4 = \frac{(q - 1)^2}{4}$.  This completes the proof.
\end{proof}

\section{The boomerang uniformity of $x^{q+2}$ over $\mathbb{F}_{q^2}$}\label{section3}
In this section, we analyze the boomerang uniformity of the power function \( f(x)=x^{q+2} \) for the cases where \( q \equiv 1 \) or \( 3\ ({\rm{mod}}\ 6)\), excluding the case where \( p = 5 \) and \( m \) is even.

For any \( b \in \mathbb{F}_{q^2} \), we put
\[
\beta(b) := \beta_f\left(1, \frac{1}{4}b\right) = \#\left\{(x, y) \in \mathbb{F}_{q^2}^2:\  
\begin{cases} 
f(x) - f(y) = \frac{1}{4}b \\ 
D_1f(x) = D_1f(y) 
\end{cases}
\right\}.
\]

By making the substitutions \( x = u - \frac{1}{2} \) and \( y = v - \frac{1}{2} \), and applying equation (\ref{20240628equation1}), we can see that \( \beta(b) \) corresponds to the number of solutions \( (u, v) \in \mathbb{F}_{q^2}^2 \) for the following system of equations:
\begin{equation}\label{20241214equation1}
  \begin{cases}(u-\frac{1}{2})^{q+2}-(v-\frac{1}{2})^{q+2}=\frac{1}{4}b,\\2u^{q+1}+u^2=2v^{q+1}+v^2.\end{cases}
\end{equation}
Expanding the first equation yields
$$(u^{q+2}+\frac{1}{4}u^q+\frac{1}{2}u)-(u^{q+1}+\frac{1}{2}u^2)-(v^{q+2}+\frac{1}{4}v^q+\frac{1}{2}v)+(v^{q+1}+\frac{1}{2}v^2)=\frac{1}{4}b.$$
Using the second equation to eliminate equal terms, we deduce that the system (\ref{20241214equation1}) is equivalent to the following system:
$$\begin{cases}(4u^{q+2}+u^q+2u)-(4v^{q+2}+v^q+2v)=b,\\2u^{q+1}+u^2=2v^{q+1}+v^2.\end{cases}$$

For any $x,y\in\mathbb{F}_q$, we have 
\begin{align*}
&4(x+yZ)^{q+2}+(x+yZ)^q+2(x+yZ)\\
=\ &(4x^3-4xy^2\alpha+3x)+(4x^2y-4y^3\alpha+y)Z,
\end{align*}
which, combined with equation (\ref{20241018equation1}), implies that for any $c,d\in\mathbb{F}_q$, $\beta(c+dZ)$ equals the number of solutions $(x,y,s,t)\in\mathbb{F}_q^4$ to the following system of equations:
\begin{equation}\label{20241214equation2}
  \begin{cases}(4x^3-4xy^2\alpha+3x)-(4s^3-4st^2\alpha+3s)=c,\\
    (4x^2y-4y^3\alpha+y)-(4s^2t-4t^3\alpha+t)=d,\\
    3x^2-y^2\alpha=3s^2-t^2\alpha,\\
    xy=st.
    \end{cases}
\end{equation}

Let ${\rm{1}}_{\{0\}}$ be the indicator function of the set $\{0\}$. For any $c,d\in\mathbb{F}_q$ and $i,j\in\{0,1\}$, let $\beta_{ij}(c+dZ)$ denote the number of solutions $(x,y,s,t)\in\mathbb{F}_q^4$ to the system (\ref{20241214equation2}) such that ${\rm{1}}_{\{0\}}(s)=i$ and ${\rm{1}}_{\{0\}}(y)=j$. For any $c,d\in\mathbb{F}_q$, we have
\begin{align}\label{20241216equation1}
  \beta_{11}(c+dZ)&=\#\left\{(x,t)\in\mathbb{F}_q^2:\ \begin{cases}
    4x^3+3x=c\\
   4t^3\alpha-t=d\\
    -3x^2=t^2\alpha
  \end{cases}\right\}
\end{align}
and 
\begin{align*}
  \beta_{10}(c+dZ)&=\#\left\{(x,y,t)\in\mathbb{F}_q^3:\ \begin{cases}
    4x^3-4xy^2\alpha+3x=c\\
    (4x^2y-4y^3\alpha+y)+4t^3\alpha-t=d\\
    3x^2-y^2\alpha=-t^2\alpha\\
    xy=0\\
    y\ne 0
  \end{cases}\right\}\\
  &=\#\left\{(y,t)\in\mathbb{F}_q^*\times\mathbb{F}_q:\ \begin{cases}
    0=c\\
    -4y^3\alpha+y+4t^3\alpha-t=d\\
    y^2=t^2
  \end{cases}\right\}\\
  &=\#\left\{(y,t)\in\mathbb{F}_q^*\times\mathbb{F}_q:\ \begin{cases}
    0=c\\
    -4y^3\alpha+y+4t^3\alpha-t=d\\
    y=\pm t
  \end{cases}\right\}.
\end{align*}
Then it is clear that \(\beta_{10}(c + dZ) = 0\) if \(c \neq 0\), and for any \(d \in \mathbb{F}_q^*\), we have the following inequality:
\begin{align}\label{20241216equation5}
  \beta_{10}(2dZ) &= \#\left\{y \in \mathbb{F}_q^*:\ -4\alpha y^3 + y = d\right\} \le 3.
\end{align}

In a similar manner, we can express \(\beta_{01}(c + dZ)\) as follows:
$$\beta_{01}(c + dZ) = \#\left\{(x, s) \in \mathbb{F}_q \times \mathbb{F}_q^* :\ \begin{cases}
  4x^3 + 3x - (4s^3 + 3s) = c \\
  d = 0 \\
  3x^2 = 3s^2
\end{cases}\right\}.$$
Then it is clear that \(\beta_{01}(c + dZ) = 0\) if \(d \neq 0\). Furthermore, if \(p > 3\), then we have
\begin{equation}\label{20241217equation2}
  \beta_{01}(2c) = \#\left\{x \in \mathbb{F}_q^*:\ 4x^3 + 3x = c\right\} \le 3
\end{equation}
for any \(c \in \mathbb{F}_q^*\). On the other hand, if \(p = 3\), then
\begin{align}
  \beta_{01}(c) &= \#\left\{(x, s) \in \mathbb{F}_q \times \mathbb{F}_q^* :\  x - s = \sqrt[3]{c}\right\} = q - 1\label{20241217equation4}
\end{align}
for any \(c \in \mathbb{F}_q^*\).

\subsection{The case $q\equiv 1\ ({\rm{mod}}\ 6)$}

In this subsection, we assume that $q\equiv 1\ ({\rm{mod}}\ 6)$. Then $-3$ is a square element in $\mathbb{F}_q$.

We begin by examining \(\beta_{11}(c + dZ)\). If there exist elements \(x_0, t_0 \in \mathbb{F}_q^*\) such that \(-3x_0^2 = t_0^2 \alpha\), then it follows that \(-3 = \frac{t_0^2 \alpha}{x_0^2} \in C_1\), which leads to a contradiction. Therefore, according to equation (\ref{20241216equation1}), we conclude that \(\beta_{11}(c + dZ) = 0\) if \(c + dZ \neq 0\).

From our previous discussion, we know that for any \(c, d \in \mathbb{F}_q^*\), the following conclusions holds \(\beta_{10}(c) = 0\), \(\beta_{01}(c) \leq 3\), \(\beta_{10}(dZ) \leq 3\), \(\beta_{01}(dZ) = 0\), and \(\beta_{10}(c + dZ) = \beta_{01}(c + dZ) = 0\). Therefore, for any \(c, d \in \mathbb{F}_q\) with \(c + dZ \neq 0\), we have 
\begin{equation}\label{20241216equation2}
  \beta_{10}(c+dZ)+\beta_{01}(c+dZ)\begin{cases}
  =0, & \text{if }c\ne 0\ \text{and }d\ne 0,\\
  \le 3, & \text{if }cd=0.
\end{cases}
\end{equation}

Finally, we consider \(\beta_{00}(c + dZ)\). We first investigate the existence of elements \(x, t \in \mathbb{F}_q\) and \(y, s \in \mathbb{F}_q^*\) such that the following system of equations holds:
\[
\begin{cases}
  3x^2 - y^2\alpha = 3s^2 - t^2\alpha, \\
  xy = st,
\end{cases}
\]
This system is equivalent to the following system:
\[
\begin{cases}
  -3s^2\left(\frac{x^2}{s^2} - 1\right) = \alpha y^2\left(\frac{t^2}{y^2} - 1\right), \\
  \frac{x}{s} = \frac{t}{y},
\end{cases}
\]
From this, we can conclude that either \(\frac{x}{s} = \frac{t}{y} \in \{\pm 1\}\) or \(-3s^2 = \alpha y^2\). However, the latter case contradicts Lemma \ref{20240620lemma6}. Therefore, we must have \(\frac{x}{s} = \frac{t}{y} \in \{\pm 1\}\). 

Consequently, for any \(c, d \in \mathbb{F}_q\) with \(c + dZ \neq 0\), we have
\begin{align*}
  \beta_{00}(2c+2dZ)&=\#\left\{(x,y)\in{\mathbb{F}_q^*}^2:\ \begin{cases}4x^3-4xy^2\alpha+3x=c\\
    4x^2y-4y^3\alpha+y=d
    \end{cases}\right\}\\
    &=\#\left\{(x,y)\in{\mathbb{F}_q^*}^2:\ \begin{cases}4x\cdot\norm(x+yZ)=c-3x\\
      4y\cdot\norm(x+yZ)=d-y
      \end{cases}\right\}.
\end{align*}
For any $d\in\mathbb{F}_q^*$, we have 
\begin{align}
  \beta_{00}(2dZ)&=\#\left\{(x,y)\in{\mathbb{F}_q^*}^2:\ \begin{cases}\norm(x+yZ)=-\frac{3}{4}\\
    4y\cdot\norm(x+yZ)=d-y
    \end{cases}\right\}\notag\\
    &=\#\left\{(x,y)\in{\mathbb{F}_q^*}^2:\ \begin{cases}\norm(x+yZ)=-\frac{3}{4}\\
      y=-\frac{d}{2}
      \end{cases}\right\}\notag\\
      &=\#\left\{x\in\mathbb{F}_q^*:\ x^2=\frac{d^2\alpha-3}{4}\right\}\le 2.\label{20241216equation3}
\end{align}
Similarly, for any $c\in\mathbb{F}_q^*$, we have 
\begin{equation}\label{20241216equation4}
  \beta_{00}(2c)=\#\left\{y\in\mathbb{F}_q^*:\ y^2=\frac{c^2+1}{4\alpha}\right\}\le 2.
\end{equation}

Now assume that $c,d\in\mathbb{F}_q^*$. Then 
\begin{align}
  &\beta_{00}(2c+2dZ)\notag\\
  =\ &\#\left\{(x,y)\in{\mathbb{F}_q^*}^2:\ \begin{cases}4x\cdot\norm(x+yZ)=c-3x\\
      4y\cdot\norm(x+yZ)=d-y
      \end{cases}\right\}\notag\\
       =\ &\#\left\{(x,y)\in{\mathbb{F}_q^*}^2:\ \begin{cases}4x\cdot\norm(x+yZ)=c-3x\\
      \frac{x}{y}=\frac{c-3x}{d-y}\\
        d\ne y
                \end{cases}\right\}\notag\\
               =\ &\#\left\{(x,y)\in{\mathbb{F}_q^*}^2:\ \begin{cases}4x^3-4xy^2\alpha+3x=c\\
        (d+2y)x=cy\\
          d\ne y
                   \end{cases}\right\}\notag\\ 
        =\ &\#\left\{y\in\mathbb{F}_q^*:\ \begin{cases}4(\frac{cy}{d+2y})^3-4(\frac{cy}{d+2y})y^2\alpha+3(\frac{cy}{d+2y})=c\\
            d\ne y,\ d+2y\ne 0
                     \end{cases}\right\}\notag\\
        =\ &\#\left\{z\in\mathbb{F}_q^*:\ \begin{cases}
          4c^2z^3-4\alpha z^3(\frac{d}{1-2z})^2+3z=1\\
          z\ne\frac{1}{2},\ z\ne\frac{1}{3}
        \end{cases}\right\}\ \ (z=\frac{y}{d+2y})\notag\\
        =\ &\#\left\{z\in\mathbb{F}_q:\ 16c^2z^5-16c^2z^4+4(c^2-\alpha d^2+3)z^3-16z^2+7z-1=0\right\}. 
        &&\label{20241217equation10}     
\end{align}

By equations (\ref{20241216equation2}), (\ref{20241216equation3}), (\ref{20241216equation4}) and (\ref{20241217equation10}), we have the following conclusion.

\begin{proposition}
  If $q\equiv 1\ ({\rm{mod}}\ 6)$, then $\beta_f\le 5$.
\end{proposition}

We also have the following observation.
\begin{lemma}
  If $q\equiv 1\ ({\rm{mod}}\ 6)$, then for any $c\in\mathbb{F}_q^*$, we have 
  \begin{align*}
    \beta(2c)\in\begin{cases}
      \{3\}, & \text{if }c^2+1\in C_1,\\
      \{0,3\}, & \text{if }c^2+1\in C_0,\\
      \{0,1,2\}, & \text{if }c^2+1=0.
    \end{cases}
  \end{align*}
\end{lemma}
\begin{proof}
  By equations (\ref{20241217equation2}) and (\ref{20241216equation4}), we have 
  \begin{align*}
    \beta(2c)&=\beta_{01}(2c)+\beta_{00}(2c)\\
    &=\#\left\{x\in\mathbb{F}_q^*:\ x^3+\frac{3}{4}x-\frac{c}{4}=0\right\}+\#\left\{y\in\mathbb{F}_q^*:\ y^2=\frac{c^2+1}{4\alpha}\right\}
  \end{align*}
  for any $c\in\mathbb{F}_q$.

  It is not difficult  to notice that for any \( c \in \mathbb{F}_q^* \), \(\beta_{00}(2c) = 2\) if and only if \( c^2 + 1 \in C_1\), and \(\beta_{00}(2c) = 0\) if and only if \( c^2 + 1 \in C_0\). 

Additionally, note that
\[
-4\left(\frac{3}{4}\right)^3 - 27\left(\frac{c}{4}\right)^2 = \frac{-27(c^2 + 1)}{4^2}.
\]
If \( c^2 + 1 \neq 0 \), then by Theorem \ref{20241219them10}, we have \(\beta_{01}(2c) = 1\) if and only if \( c^2 + 1 \in C_1\). Consequently, if \( c^2 + 1 \in C_1\), then \(\beta(2c) = 3\). On the other hand, if \( c^2 + 1 \in C_0\), then \(\beta(2c) \in \{0, 3\}\).

Finally, if \( c^2 + 1 = 0\), then \(\beta_{00}(2c) = 0\) and \(\beta_{01}(2c) \leq 2\), as indicated in the remark of Theorem \ref{20241219them10}, which implies that \(\beta(2c) \in \{0, 1, 2\}\). This completes the proof.

\end{proof}

We begin by considering the case where \( p > 5 \). The following simple lemma will be utilized in the proof of Theorem \ref{20241217them2}.

\begin{lemma}\label{20241216lemma4}
  If the polynomial \( x^6 + bx^3 + c \in \mathbb{F}_q[x] \) with \( c \neq 0 \) is the square of a polynomial, then it must hold that \( b^2 - 4c = 0 \).
\end{lemma}
\begin{proof}
  Assume that 
  \[
  x^6 + bx^3 + c = (x^3 + ex^2 + fx + g)^2
  \]
  for some elements \( e, f, g \) in the algebraic closure of \( \mathbb{F}_q \). By comparing the coefficients of the constant term, \( x \), \( x^5 \), and \( x^6 \) on both sides of the equation, we find that \( e = f = 0 \). This leads us to the following simplification:
  \[
  x^6 + bx^3 + c = (x^3 + g)^2 = x^6 + 2gx^3 + g^2.
  \]
  From this, we can conclude that \( b^2 - 4c = (2g)^2 - 4g^2 = 0 \).   This completes the proof.
\end{proof}

The result presented below is the primary outcome of this subsection.

\begin{theorem}\label{20241217them2}
  If \( p > 5 \), \( q \equiv 1 \pmod{6} \), and \( q \not\in \{ 7, 19, 313 \} \), then \( \beta_f = 5 \).
\end{theorem}

\begin{proof}
Recall that \(\alpha\) is an arbitrary non-square element in \(\mathbb{F}_q\). Hence we only need to prove that there exists \(\alpha \in C_1\) such that \(\beta(8\alpha Z) = 5\).

By equations (\ref{20241216equation5}) and (\ref{20241216equation3}), we have
\begin{align*}
  \beta(8\alpha Z) &= \beta_{10}(8\alpha Z) + \beta_{00}(8\alpha Z)\\
  &= \#\{ y \in \mathbb{F}_q^* :\ y^3 - \frac{1}{4\alpha}y + 1 = 0 \} + \#\{ x \in \mathbb{F}_q^*:\ x^2 = \frac{16\alpha^3 - 3}{4} \}.
\end{align*}
Thus, \(\beta(8\alpha Z) = 5\) if and only if \(16\alpha^3 - 3 \in C_0\) and the cubic equation \(y^3 - \frac{1}{4\alpha}y + 1 = 0\) has three distinct solutions in \(\mathbb{F}_q^*\). 

Let \(\alpha' = \frac{1}{4\alpha}\). Then it suffices to show that there exists \(\alpha' \in C_1\) such that \(1 - 12\alpha'^3 \in C_1\) and that the cubic equation \(y^3 - \alpha'y + 1 = 0\) has three distinct solutions in \(\mathbb{F}_q^*\). According to Theorem \ref{20241219them10}, this cubic equation has three distinct solutions in \(\mathbb{F}_q^*\) if \(4\alpha'^3 - 27 \in C_0\) and \(\frac{1}{2}(-1 + c\omega)\) is a cube in \(\mathbb{F}_q\), where \(\omega\) is a square root of \(-3\) in \(\mathbb{F}_q\) and \(c \in \mathbb{F}_q^*\) is such that \(4\alpha'^3 - 27 = 81c^2\).

We claim that it suffices to find an \(x \in \mathbb{F}_q\) such that \(2x^3 + 1 \neq 0\), \(x^6 + x^3\) is a cube in \(\mathbb{F}_q\), \(x^6 + x^3 \in C_1\), and \(2^2 3^4 (x^6 + x^3) + 1 \in C_1\). Indeed, if such an \(x\) exists, we can take \(\alpha' \in \mathbb{F}_q\) to be a cubic root of \(-27(x^6 + x^3)\). Since \(-3 \in C_0\) and \(x^6 + x^3 \in C_1\), it follows that \(\alpha' \in C_1\). Moreover, we have
\[
1 - 12\alpha'^3 = 1 + 12 \cdot 27(x^6 + x^3) = 2^2 3^4 (x^6 + x^3) + 1 \in C_1,
\]
and
\[
4\alpha'^3 - 27 = -27(4x^6 + 4x^3 + 1) = -27(2x^3 + 1)^2 \in C_0.
\]
We can take \(c = \frac{2x^3 + 1}{\omega}\), ensuring that \(\frac{1}{2}(-1 + c\omega) = x^3\) is a cube in \(\mathbb{F}_q\). Thus, the claim is proved.

  We will prove that such an \( x \) exists. Let \( N \) be the number of such \( x \), and let \( \chi \) be a multiplicative character of \( \mathbb{F}_q \) of order $6$. We define \( \eta = \chi^3 \), which is the quadratic character of \( \mathbb{F}_q \), and \( \psi = \chi^2 \), which is a multiplicative character of \( \mathbb{F}_q \) with an order of 3. 

Define the following two polynomials:
\[
f(x) = x^6 + x^3,\qquad g(x) = x^6 + x^3 + \frac{1}{2^2 3^4}.
\]
Then, utilizing the basic properties of multiplicative characters, we have
  \begin{align*}
    12N&=\sum\limits_{x\in\mathbb{F}_q\setminus\Omega}\sum\limits_{i=0}^2\psi^i\big(f(x)\big)\sum\limits_{j=0}^1(-1)^j\eta^j\big(f(x)\big)\sum\limits_{k=0}^1(-1)^k\eta^k\big(g(x)\big)\\
    &=\sum\limits_{i=0}^2\sum\limits_{j=0}^1\sum\limits_{k=0}^1(-1)^{j+k}\sum\limits_{x\in\mathbb{F}_q\setminus\Omega}\chi\Big(\big(f(x)\big)^{2i+3j}\big(g(x)\big)^{3k}\Big)\\
    &=\sum\limits_{i=0}^2\sum\limits_{j=0}^1\sum\limits_{k=0}^1(-1)^{j+k}S_{i,j,k}-T,
  \end{align*}
  where $\Omega=\{x\in\mathbb{F}_q:\ 2x^3+1=0\ \mbox{or}\ x^6+x^3=0\ \mbox{or}\ 2^23^4(x^6+x^3)+1=0\}$,
  $$T=\sum\limits_{i=0}^2\sum\limits_{j=0}^1\sum\limits_{k=0}^1(-1)^{j+k}\sum\limits_{x\in\Omega}\chi\Big(\big(f(x)\big)^{2i+3j}\big(g(x)\big)^{3k}\Big),$$
  and
  $$S_{i,j,k}=\sum\limits_{x\in\mathbb{F}_q}\chi\Big(\big(f(x)\big)^{2i+3j}\big(g(x)\big)^{3k}\Big)$$
  for any $i\in\{0,1,2\}$ and $j,k\in\{0,1\}$. It is clear that $\#\Omega\le 3+4+6=13$ and thus
  $$|T|\le 3\cdot 2\cdot 2\cdot\#\Omega\le 156.$$
  
  Now we consider the character sums \( S_{i,j,k} \). We have \( S_{0,0,0} = q \). By Lemma \ref{20241216lemma4}, the polynomial \( g(x) \) is the square of a polynomial if and only if \( 1 - 4 \cdot \left( \frac{1}{2^2 3^4} \right)^2 = 0 \). This simplifies to the condition \( 5 \equiv 0 \ (\text{mod}\ p) \). Since \( p > 5 \), this cannot occur.  Therefore, \( g(x)^3 \) is not the sixth power of a polynomial.

Moreover, note that 
\[
f_{i,j,k}(x) := \left( f(x) \right)^{2i + 3j} \left( g(x) \right)^{3k} = x^{3(2i + 3j)} (x^3 + 1)^{2i + 3j} g(x)^{3k},
\]
where \( x \), \( x^3 + 1 \), and \( g(x) \) are pairwise coprime. Consequently, the multiplicity of the factor \( x^3 + 1 \) in \( f_{i,j,k}(x) \) is \( 2i + 3j \), which cannot be a multiple of 6 for any \( i \in \{ 0, 1, 2 \} \) and \( j \in \{ 0, 1 \} \) with \( i + j > 0 \). 

Thus, for any \( i \in \{ 0, 1, 2 \} \) and \( j, k \in \{ 0, 1 \} \) with \( i + j + k > 0 \), \( f_{i,j,k}(x) \) is not the sixth power of a polynomial. By Theorem \ref{20240904them3}, we have

  $$|S_{i,j,k}|\le\begin{cases} 
    (1+3-1)\sqrt{q}=3\sqrt{q},&\mbox{if }i+j>0\ \mbox{and }k=0,\\
    (6-1)\sqrt{q}=5\sqrt{q},&\mbox{if }i+j=0\ \mbox{and }k=1,\\
    (1+3+6-1)\sqrt{q}=9\sqrt{q},&\mbox{if }i+j>0\ \mbox{and }k=1.
  \end{cases}
    $$

It follows that
$$12N \ge q - 156 - (5 + 2 \cdot 9 + 2 \cdot 3) \sqrt{q} = q - 29\sqrt{q} - 156.$$
Thus, if \( q \ge 1132 \), it can be concluded that \( N > 0 \). For the case where \( 7 \le q < 1132 \), we directly verify the existence of such an \( x \) using a Python program. We find that the only counterexamples are \( 7, 13, 19, 31, 37, 43, 79, 139, \) and \( 313 \). 

For these counterexamples, we utilize the Python program to calculate \( \beta_f \) directly and identify \( 7, 19, \) and \( 313 \) as the only counterexamples. The proof is, therefore, completed.
\end{proof}

\begin{remark}
  If $q=7$, then $\beta_f=3$; if $q=19$, then $\beta_f=4$; and if $q=313$, then $\beta_f=3$.
\end{remark}

The following observation clarifies why the case \( p = 5 \) cannot be addressed using the method employed to prove Theorem \ref{20241217them2}.

\begin{lemma}
If \( p = 5 \) and \( m \) is even, then for any \( d \in \mathbb{F}_q^* \), we have 
\[
\beta(2dZ) \in \begin{cases}
  \{3\}, & \text{if } d^2\alpha + 2 \in C_0, \\
  \{0, 3\}, & \text{if } d^2\alpha + 2 \in C_1.
\end{cases}
\]
\end{lemma}

\begin{proof}
By equations (\ref{20241216equation5}) and (\ref{20241216equation3}), we have
\begin{align*}
  \beta(2dZ)&= \beta_{10}(2dZ) + \beta_{00}(2dZ) \\
&= \#\left\{ y \in \mathbb{F}_q^*:\ y^3 + \frac{1}{\alpha}y - \frac{d}{\alpha} = 0 \right\} + \#\left\{ x \in \mathbb{F}_q^*:\ x^2 = \frac{d^2\alpha + 2}{4} \right\}
\end{align*}
for any \( d \in \mathbb{F}_q^* \). Since \( m \) is even, we have \( 3 \in C_0 \), which implies that \( d^2\alpha + 2 \neq 0 \) for any \( d \in \mathbb{F}_q^* \). Therefore, \( \beta_{00}(2dZ) \) can be either \( 0 \) or \( 2 \) for any \( d \in \mathbb{F}_q^* \), and \( \beta_{00}(2dZ) = 2 \) if and only if \( d^2\alpha + 2 \in C_0 \). By Theorem \ref{20241219them10}, for any \( d \in \mathbb{F}_q^* \), \( \beta_{10}(2dZ) = 1 \) if and only if \( -4\left(\frac{1}{\alpha}\right)^3 - 27\left(-\frac{d}{\alpha}\right)^2 = \frac{3(d^2\alpha + 2)}{\alpha^3} \in C_1 \), i.e., \( d^2\alpha + 2 \in C_0 \). Thus, if \( d^2\alpha + 2 \in C_0 \), then \( \beta(2dZ) = 3 \), and if \( d^2\alpha + 2 \in C_1 \), then \( \beta(2dZ) = \beta_{10}(2dZ) \in \{0, 3\} \), which completes the proof.
\end{proof}

\begin{remark}
We conjecture that if \( p = 5 \) and \( m \ge 4 \) is even, then \( \beta_f \) is also \( 5 \).
\end{remark}

\subsection{The case $p=3$}
In this subsection, we assume that $p=3$. By equation (\ref{20241216equation1}), for any $c,d\in\mathbb{F}_q$, we have
\begin{align}
  \beta_{11}(c+dZ)&=\#\left\{(x,t)\in\mathbb{F}_q^2:\ \begin{cases}
    x^3=c\\
   t^3\alpha-t=d\\
    t=0
  \end{cases}\right\}\notag\\
  &=\#\left\{x\in\mathbb{F}_q:\ \begin{cases}
    x=\sqrt[3]{c}\\
   0=d\\
  \end{cases}\right\}=\begin{cases}
    1, & \text{if }d=0,\\
    0, & \text{if }d\ne 0.
  \end{cases}\label{20241217equation5}
\end{align}

By equation (\ref{20241216equation5}), for any $d\in\mathbb{F}_q^*$, we have
\begin{align*}
  \beta_{10}(2dZ)&=\#\left\{y\in\mathbb{F}_q^*:\ -\alpha y^3+y=d\right\}.
\end{align*}
Consider the map \(\mathbb{F}_q \rightarrow \mathbb{F}_q\) defined by \(y \mapsto -\alpha y^3 + y\). This map is \(\mathbb{F}_p\)-linear and its kernel is given by \(\{y \in \mathbb{F}_q:\ y = 0 \text{ or } y^2 = \frac{1}{\alpha}\}\). Since $\alpha\in C_1$, the kernel is $\{0\}$, and thus the map is bijective. Consequently, we find that \(\beta_{10}(2dZ) = 1\) for any \(d \in \mathbb{F}_q^*\). Additionally, by equation (\ref{20241217equation4}), for any \(c \in \mathbb{F}_q^*\), it follows that \(\beta_{01}(c) = q - 1\).

 Finally, from equation (\ref{20241214equation2}), for any \(c, d \in \mathbb{F}_q\) with \(c + dZ \neq 0\), we have
\begin{align*}
  &\beta_{00}(c+dZ)\\
  =\ &\#\left\{(x,y,s,t)\in\mathbb{F}_q^4:\ \begin{cases}(x^3-xy^2\alpha)-(s^3-st^2\alpha)=c\\
    (x^2y-y^3\alpha+y)-(s^2t-t^3\alpha+t)=d\\
    y^2=t^2\\
    xy=st\\
    y,s\ne 0
    \end{cases}\right\}\\
    =\ &\#\left\{(x,y,s,t)\in\mathbb{F}_q^4:\ \begin{cases}(x^3-xy^2\alpha)-(s^3-st^2\alpha)=c\\
      (x^2y-y^3\alpha+y)-(s^2t-t^3\alpha+t)=d\\
      y=-t\ne 0\\
      x=-s\ne 0
      \end{cases}\right\} \\
            =\ &\#\left\{(x,y)\in{\mathbb{F}_q^*}^2:\ \begin{cases}x^3-xy^2\alpha=2c\\
        x^2y-y^3\alpha+y=2d
              \end{cases}\right\}\\
              =\ &\#\left\{(x,y)\in{\mathbb{F}_q^*}^2:\ \begin{cases}
        x\cdot \norm(x+yZ)=2c\\
      y\cdot \norm(x+yZ)=2(d+y)
              \end{cases}\right\}\notag.
\end{align*}
Then it is clear that for any $d\in\mathbb{F}_q^*$, we have $\beta_{00}(dZ)=0$, and for any $c\in\mathbb{F}_q^*$, we have
\begin{align}
  \beta_{00}(c)&=\#\left\{(x,y)\in{\mathbb{F}_q^*}^2:\ \begin{cases}
    x\cdot \norm(x+yZ)=2c\\
  y\cdot \norm(x+yZ)=2y
          \end{cases}\right\}\notag\\
  &=\#\left\{(x,y)\in{\mathbb{F}_q^*}^2:\ \begin{cases}
    x=c\\
 \norm(x+yZ)=2
          \end{cases}\right\}\notag\\
  &=\#\left\{y\in\mathbb{F}_q^*:\ \norm(c+yZ)=2\right\}\notag\\
  &=\#\left\{y\in\mathbb{F}_q^*:\ y^2=\frac{c^2+1}{\alpha}\right\}=\begin{cases}
    2,&c^2+1\in C_1,\\
    0,&\mbox{otherwise}.
  \end{cases}\label{20241217equation7}
\end{align} 

Now assume that $c,d\in\mathbb{F}_q^*$. Then
\begin{align}
  &\beta_{00}(c+dZ)\notag\\
  =\ &\#\left\{(x,y)\in{\mathbb{F}_q^*}^2:\ \begin{cases}
    x\cdot \norm(x+yZ)=2c\\
  y\cdot \norm(x+yZ)=2(d+y)
          \end{cases}\right\}\notag\\
  =\ &\#\left\{(x,y)\in{\mathbb{F}_q^*}^2:\ \begin{cases}
    x\cdot \norm(x+yZ)=2c\\
    cy=x(d+y)
            \end{cases}\right\}\notag\\
            =\ &\#\left\{y\in\mathbb{F}_q^*\setminus\{-d\}:\ (\frac{cy}{d+y})^3-\frac{cy^3\alpha}{d+y}=2c\right\}\notag\\
            =\ &\#\left\{z\in\mathbb{F}_q^*\setminus\{1\}:\ z^5-\frac{d^2\alpha}{c^2}z^3-(1+\frac{1}{c^2})z^2+\frac{d^2\alpha}{c^2}=0\right\}\ \ (z=\frac{d}{d+y})\notag\\
            =\ &\#\left\{z\in\mathbb{F}_q:\ z^5-A(c,d)z^3-\Big(1+B(c)\Big)z^2+A(c,d)=0\right\},\label{20241214equation8}
\end{align}
where $A(c,d)=\frac{d^2\alpha}{c^2}$ and $B(c)=\frac{1}{c^2}$. Note that as $c,d$ vary in $\mathbb{F}_q^*$, $A(c,d)$ ($B(c)$, resp.) can represent any non-square (square, resp.) element in $\mathbb{F}_q^*$.

\begin{proposition}\label{20241214prop1}
  \begin{enumerate}[(1)]
    \item For any $c\in\mathbb{F}_q^*$, we have$$\beta(c)=\begin{cases}
  q+2,&\mbox{if }c^2+1\in C_1,\\
  q,&\mbox{otherwise}.
\end{cases}$$
\item For any $d\in\mathbb{F}_q^*$, we have $\beta(dZ)=1$. In particular, $\nu_1>0$.
\item For any $c,d\in\mathbb{F}_q^*$, we have $\beta(c+dZ)\in\{0,1,2,5\}$.
  \end{enumerate}
\end{proposition}

\begin{proof}

The points (1) and (2) follow directly  from the previous discussion. Now we prove the point (3). From the preceding analysis, for \(c,d \in \mathbb{F}_q^*\), it follows that \(\beta(c + dZ) = \beta_{00}(c + dZ)\). According to equation (\ref{20241214equation8}), we need to demonstrate that if \(A \in C_1\) and \(B \in C_0\), then the number \(N(f)\) of the roots of the polynomial \(f(z) = z^5 - Az^3 - (1 + B)z^2 + A\) in \(\mathbb{F}_q\) can only be 0, 1, 2, or 5.

We use Mathematica to compute the discriminant \(D(f)\) of the polynomial \(f(z)\), which yields \(A^4B^3\). Since \(B \in C_0\), we know that \(\eta\big(D(f)\big) = 1\). Let \(k\) represent the number of monic irreducible factors of \(f(z)\) over \(\mathbb{F}_q\). By Theorem \ref{20241219them12}, we establish that \(k\) must be an odd number. Therefore, the possible types of \(f\) are \((5)\), \((1,2,2)\), \((1,1,3)\), and \((1,1,1,1,1)\), which correspond to the cases where \(N(f) = 0\), \(1\), \(2\), and \(5\), respectively. The proof is thus completed.
\end{proof}

The following theorem is the main result of this subsection.

\begin{theorem}
  If $p=3$, then $\beta_f=q+2$. If $q\ge 9$, then 
  $$\nu_{q+2}=\begin{cases}
    \frac{q-1}{2},&\mbox{if }n\ \mbox{is even},\\
    \frac{q+1}{2},&\mbox{if }n\ \mbox{is odd},
  \end{cases}$$
  and
  $$\nu_q=\begin{cases}
    \frac{q-1}{2},&\mbox{if }n\ \mbox{is even},\\
    \frac{q-3}{2},&\mbox{if }n\ \mbox{is odd}.
  \end{cases}$$
\end{theorem}

\begin{proof}
  By Proposition \ref{20241214prop1} (1) and Theorem \ref{20241218them11}, we have 
  $$\nu_{q+2}=2\cdot(0,1)=\begin{cases}
    \frac{q-1}{2},&\mbox{if }q\equiv 1\ ({\rm{mod}}\ 4),\\
    \frac{q+1}{2},&\mbox{if }q\equiv 3\ ({\rm{mod}}\ 4),
  \end{cases}$$
  and 
  $$\nu_q=2\cdot(0,0)+\{c\in\mathbb{F}_q:\ c^2=-1\}=\begin{cases}
    \frac{q-1}{2},&\mbox{if }q\equiv 1\ ({\rm{mod}}\ 4),\\
    \frac{q-3}{2},&\mbox{if }q\equiv 3\ ({\rm{mod}}\ 4).
  \end{cases}$$
\end{proof}

We prove the following additional result on the boomerang spectrum of $f$.

\begin{proposition}
  If $q\ge 9$, then $\nu_2>0$; and if $q\ge 27$, then $\nu_5>0$.
\end{proposition}

\begin{proof}

If \( q \geq 81 \), then by Lemma \ref{20241219lemma7}, for any \( i \in \{0, 1\} \), there exists \( x_i \in C_1 \) such that \( x_i^3 - x_i = b_i^2 \) for some \( b_i \in \mathbb{F}_q^* \) with \( \text{Tr}_{\mathbb{F}_q/\mathbb{F}_3}(b_i) = i \). For each \( i\in\{0,1\} \), we choose \( y_i \) from the set \( \{x_i, x_i + 1, x_i - 1\} \) such that the following conditions are satisfied: \( \eta(y_i) = \eta(-1) \), \( \eta(y_i - 1) = -\eta(-1) \), and \( \eta(y_i + 1) = -1 \). We then define \( A_i = -\frac{y_i + 1}{y_i} \) and \( B_i = (A_i - 1)^2 (A_i + 1) \).

Then, it is easy to verify that
\[
\eta(A_i) = -1, \quad \eta(A_i + 1) = 1, \quad \eta(A_i - 1) = -1,
\]
which implies that \( B_i \in C_0 \). We observe the following:
   \begin{align*}
    f^{(i)}(z)&:=z^5-A_iz^3-(1+B_i)z^2+A_i\\
    &=\big(z^2-(A_i+1)z+1\big)\big(z^3+(A_i+1)z^2+A_i(A_i+1)z+A_i\big).
  \end{align*}
  Consider the quadratic polynomial $f_1^{(i)}(z):=z^2-(A_i+1)z+1$. Since its discriminant $(A_i+1)^2-1=A_i(A_i-1)$ belongs to $C_0$, it has two distinct roots in $\mathbb{F}_q$. 
  
  Next, consider the cubic polynomial $f_2^{(i)}(z)=z^3+(A_i+1)z^2+A_i(A_i+1)z+A_i$. If we put 
  $$g_2^{(i)}(z)=\frac{1}{A_i(1-A_i)}z^3f_2(\frac{1}{z}+A_i)=z^3+\frac{A_i+1}{A_i(1-A_i)}z+\frac{1}{A_i(1-A_i)},$$
  then $g_2^{(i)}(z)$ and $f_2^{(i)}(z)$ have the same number of roots in $\mathbb{F}_q$. Note that $-\frac{A_i+1}{A_i(1-A_i)}=\frac{A_i+1}{A_i(A_i-1)}\in C_0$ and
  $$b_i^2=y_i^3-y_i=(-\frac{1}{A_i+1})^3+\frac{1}{A_i+1}=\frac{A_i(A_i-1)}{(A_i+1)^3},$$
  which implies that
  $$(\frac{1}{(A_i+1)b_i})^2=\frac{A_i+1}{A_i(A_i-1)}.$$
Let \( c_i = \frac{1}{(A_i + 1) b_i} \). Then we have
\[
\text{Tr}_{\mathbb{F}_q/\mathbb{F}_3}\left(\frac{1}{A_i(1-A_i)c_i^3}\right) = -\text{Tr}_{\mathbb{F}_q/\mathbb{F}_3}(b_i) = -i.
\]

According to Theorem \ref{20241219them7}, the polynomial \( g_2^{(0)}(z) \) has three distinct roots in \( \mathbb{F}_q \), while \( g_2^{(1)}(z) \) has no roots in \( \mathbb{F}_q \). By equation (\ref{20241214equation8}), this proposition is proven for \( q \geq 81 \). The remaining cases are verified directly using a computer program.
\end{proof}

\section{The Walsh spectrum of $x^{3^m+2}$ over $\mathbb{F}_{3^{2m}}$ and its consequences}\label{section4}

\subsection{The Walsh spectrum of $x^{3^m+2}$ over $\mathbb{F}_{3^{2m}}$}

In \cite{yan2022note}, the authors demonstrated that for any odd prime power \( q \), the differential spectrum of the power function \( g(x) = x^{2q-1} \) over \( \mathbb{F}_{q^2} \) is given by
\[
\text{DS}_g = \left\{ \omega_0 = \frac{q^2 + 2 - q}{2}, \ \omega_2 = \frac{q^2 - q}{2}, \ \omega_q = 1 \right\}.
\]
It can be easily observed that when $p=3$, our power function $f(x)=x^{q+2}$ over $\mathbb{F}_{q^2}$, which is linearly equivalent to \( x^{2q + 1} \), has the same differential spectrum as $g$.

Furthermore, in \cite{li2022class}, the authors analyzed the value distribution of the Walsh spectrum of \( g \), revealing that it is \( 4 \)-valued. This naturally raises the question of whether our power function \( f \) over \( \mathbb{F}_{q^2} \) with \( p = 3 \) shares the same property. The answer is affirmative. In the remainder of this section, we will assume that \( p = 3 \).

\begin{proposition}\label{20240623prop4}
  The Walsh spectrum of $f$ takes value in $\{-q,\ 0,\ q,\ 2q\}$.
\end{proposition}
\begin{proof}
  Let $\epsilon$ be a primitive element in $\mathbb{F}_{q^2}$ and let
  $$\lambda=\begin{cases}
      \epsilon^{\frac{q+1}{2}}, & \text{if }q\equiv 1\ ({\rm{mod}}\ 4), \\
      \epsilon^{\frac{q-1}{2}}, & \text{if }q\equiv 3\ ({\rm{mod}}\ 4).
    \end{cases}$$
  Then $\lambda$ is a non-square element in $\mathbb{F}_{q^2}$ such that
  $$\lambda^q=\begin{cases}-\lambda,           & \text{if }q\equiv 1\ ({\rm{mod}}\ 4), \\
             -\frac{1}{\lambda}, & \text{if }q\equiv 3\ ({\rm{mod}}\ 4).\end{cases}$$
  For any $a\in\mathbb{F}_{q^2}$ and $b\in\mathbb{F}_{q^2}^*$, we have
  \begin{align*}
    W_f(a,b) & =\sum\limits_{x\in\mathbb{F}_{q^2}}\xi_3^{\trace (bx^{q+2}-ax)}                                               \\
             & =1+\sum\limits_{x\in C_0}\xi_3^{\trace (bx^{q+2}-ax)}+\sum\limits_{x\in\lambda C_0}\xi_3^{\trace (bx^{q+2}-ax)}.
  \end{align*}
 Recall that the absolute Frobenius map $\mathbb{F}_{q^2}\rightarrow\mathbb{F}_{q^2}$, $x\mapsto x^3$ is an field automorphism such that $\trace(x^3)=\trace(x)$. We will use $x^{\frac{1}{3}}$ to denote the unique preimage of $x\in\mathbb{F}_{q^2}$ under the Frobenius map. By Lemma \ref{20240617lemma1}, we have
  \begin{align*}
        & \sum\limits_{x\in C_0}\xi_3^{\trace (bx^{q+2}-ax)}                                                                                                                                                                                               \\
    =\  & \frac{1}{2}\sum\limits_{(y,z)\in\mathbb{F}_{q}^*\times\mathbb{U}}\xi_3^{\trace\big(b(yz)^{q+2}-ayz\big)}                                                                                                                                       \\
    =\  & \frac{1}{2}\sum\limits_{(y,z)\in\mathbb{F}_{q}^*\times\mathbb{U}}\xi_3^{\trace(by^3z-ayz)}                                                                                                                                             \\
    =\  & \frac{1}{2}\sum\limits_{(y,z)\in\mathbb{F}_{q}^*\times\mathbb{U}}\xi_3^{\trace(by^3z)-\trace(ayz)}                                                                                                                                     \\
    =\  & \frac{1}{2}\sum\limits_{(y,z)\in\mathbb{F}_{q}^*\times\mathbb{U}}\xi_3^{\trace(b^{\frac{1}{3}}z^{\frac{1}{3}}y)-\trace(ayz)}                                                                                                                   \\
    =\  & \frac{1}{2}\sum\limits_{z\in\mathbb{U}}\sum\limits_{y\in\mathbb{F}_q^*}\xi_3^{\trace\Big((b^{\frac{1}{3}}z^{\frac{1}{3}}-az)y\Big)}                                                                                                            \\
    =\  & -\frac{1}{2}\cdot\#\mathbb{U}+\frac{1}{2}\sum\limits_{z\in\mathbb{U}}\sum\limits_{y\in\mathbb{F}_q}\xi_3^{{\rm{Tr}}_{\mathbb{F}_q/\mathbb{F}_3}\Big(y\cdot{\rm{Tr}}_{\mathbb{F}_{q^2}/\mathbb{F}_{q}}(b^{\frac{1}{3}}z^{\frac{1}{3}}-az)\Big)} \\
    =\  & -\frac{q+1}{2}+\frac{q}{2}\cdot\#\{z\in\mathbb{U}:\ \traced(b^{\frac{1}{3}}z^{\frac{1}{3}}-az)=0\}.
  \end{align*}
  Note that
  \begin{align*}
    \traced(b^{\frac{1}{3}}z^{\frac{1}{3}}-az) & =b^{\frac{1}{3}}z^{\frac{1}{3}}-az+(b^{\frac{1}{3}}z^{\frac{1}{3}}-az)^q    \\
                                                & =b^{\frac{1}{3}}z^{\frac{1}{3}}-az+b^{\frac{q}{3}}z^{-\frac{1}{3}}-a^qz^{-1}.
  \end{align*}
  Hence
  \begin{align*}
    \traced(b^{\frac{1}{3}}z^{\frac{1}{3}}-az)=0 & \iff bz-a^3z^3+b^qz^{-1}-a^{3q}z^{-3}=0 \\
                                                  & \iff a^3z^6-bz^4-b^qz^2+a^{3q}=0,
  \end{align*}
  and thus
  \begin{align*}
    \sum\limits_{x\in C_0}\xi_3^{\trace (bx^{q+2}-ax)} & =q\cdot\#\{s\in\mathbb{U}^2:\ a^3s^3-bs^2-b^qs+a^{3q}=0\}-\frac{q+1}{2},
  \end{align*}
  where $\mathbb{U}^2=\{u^2:\ u\in\mathbb{U}\}$.
  \begin{enumerate}[(1)]
    \item If $q\equiv 1\pmod{4}$, then we have
          \begin{align*}
                & \sum\limits_{x\in\lambda C_0}\xi_3^{\trace (bx^{q+2}-ax)}=\sum\limits_{u\in C_0}\xi_3^{\trace (b\lambda^3u^{q+2}-a\lambda x)} \\
            =\  & q\cdot\#\{s\in\mathbb{U}^2:\ a^3\lambda^3s^3-b\lambda^{3}s^2-b^q\lambda^{3q}s+a^{3q}\lambda^{3q}=0\}-\frac{q+1}{2}          \\
            =\  & q\cdot\#\{s\in\mathbb{U}^2:\ a^3\lambda^3s^3-b\lambda^{3}s^2+b^q\lambda^{3}s-a^{3q}\lambda^{3}=0\}-\frac{q+1}{2}            \\
            =\  & q\cdot\#\{s\in\mathbb{U}^2:\ a^3s^3-bs^2+b^qs-a^{3q}=0\}-\frac{q+1}{2}                                                      \\
            =\  & q\cdot\#\{s\in-\mathbb{U}^2:\ a^3s^3-bs^2-b^qs+a^{3q}=0\}-\frac{q+1}{2}.
          \end{align*}
          Since $(-1)^{\frac{q+1}{2}}=-1$, we have $-1\not\in\mathbb{U}^2$ and thus $\mathbb{U}^2\cap(-\mathbb{U}^2)=\emptyset$.
    \item If $q\equiv 3\pmod{4}$, then we have
          \begin{align*}
                & \sum\limits_{x\in\lambda C_0}\xi_3^{\trace (bx^{q+2}-ax)}=\sum\limits_{u\in C_0}\xi_3^{\trace (b\lambda^{-1}u^{q+2}-a\lambda x)} \\
            =\  & q\cdot\#\{s\in\mathbb{U}^2:\ a^3\lambda^3s^3-b\lambda^{-1}s^2-b^q\lambda^{-q}s+a^{3q}\lambda^{3q}=0\}-\frac{q+1}{2}            \\
            =\  & q\cdot\#\{s\in\mathbb{U}^2:\ a^3\lambda^3s^3-b\lambda^{-1}s^2+b^q\lambda s-a^{3q}\lambda^{-3}=0\}-\frac{q+1}{2}                \\
            =\  & q\cdot\#\{s\in\mathbb{U}^2:\ a^3\lambda^6s^3-b\lambda^2s^2+b^q\lambda^4s-a^{3q}=0\}-\frac{q+1}{2}                              \\
            =\  & q\cdot\#\{s\in\lambda^2\mathbb{U}^2:\ a^3s^3-bs^2+b^qs-a^{3q}=0\}-\frac{q+1}{2}                                                \\
            =\  & q\cdot\#\{s\in-\lambda^2\mathbb{U}^2:\ a^3s^3-bs^2-b^qs+a^{3q}=0\}-\frac{q+1}{2}                                               \\
            =\  & q\cdot\#\{s\in\lambda^2\mathbb{U}^2:\ a^3s^3-bs^2-b^qs+a^{3q}=0\}-\frac{q+1}{2},
          \end{align*}
          noticing that $-1\in\mathbb{U}^2$. Note that $\lambda^{q+1}=\epsilon^{\frac{q^2-1}{2}}=-1$, i.e., $\lambda\not\in\mathbb{U}$, which implies that $\mathbb{U}^2\cap(\lambda^2\mathbb{U}^2)=\emptyset$.
  \end{enumerate}

  Hence
  \begin{align*}
    W_f(a,b) & =-q+q\cdot\#\Lambda(a,b),
  \end{align*}
  where
  $$\Lambda(a,b)=\{s\in\mathbb{U}^2\cup(-\mathbb{U}^2):\ a^3s^3-bs^2-b^qs+a^{3q}=0\}$$
  if $q\equiv 1\pmod{4}$ and
  $$\Lambda(a,b)=\{s\in\mathbb{U}^2\cup(\lambda^2\mathbb{U}^2):\ a^3s^3-bs^2-b^qs+a^{3q}=0\}$$
  if $q\equiv 3\pmod{4}$. Since $a^3s^3-bs^2-b^qs+a^{3q}=0$ is a cubic equation, we have $\#\Lambda(a,b)\in\{0,\ 1,\ 2,\ 3\}$, which implies that $W_f(a,b)\in\{-q,\ 0,\ q,\ 2q\}$.
\end{proof}

\begin{proposition}\label{20240623prop2}
  We have
  $$\sum\limits_{\substack{a\in\mathbb{F}_{q^2},b\in\mathbb{F}_{q^2}^*}}\big(W_f(a,b)-1\big)^3=q^2(q^2-1)(q^3-3q^2+2)$$
\end{proposition}

\begin{proof}
  From the proof of \cite[Lemma 2.3]{li2022class}, we can see that
  $$\sum\limits_{\substack{a\in\mathbb{F}_{q^2},b\in\mathbb{F}_{q^2}^*}}\big(W_f(a,b)-1\big)^3=q^4(q^2-1)\cdot\big(\delta_f(1,1)-2\big)-q^2(q^2-1)(q^2-2),$$
  By equation (\ref{20240620equation1}) and Proposition \ref{20240620prop4}, we have
  $$\delta_f(1,1)=\delta(\frac{3}{4})=\delta(0)=q.$$
  Then, this proposition follows immediately.
\end{proof}

\begin{theorem}\label{20240630them1}
  Assume that $p=3$. When $(a,b)$ runs through $\mathbb{F}_{q^2}\times\mathbb{F}_{q^2}^*$, the value distribution of the Walsh transform of $f$ is given by
  $$W_f(a,b)=\begin{cases}
      -q, & \text{occurs}\ \frac{q^4-q^3-q^2+q}{3}\ \text{times}, \\
      0,  & \text{occurs}\ \frac{q^4-q^3-q^2+q}{2}\ \text{times}, \\
      q,  & \text{occurs}\ q^3-q\ \text{times},                   \\
      2q, & \text{occurs}\ \frac{q^4-q^3-q^2+q}{6}\ \text{times}.
    \end{cases}$$
\end{theorem}

\begin{proof}
  For $i\in\{-q,0,q,2q\}$, let
  $$\eta_i=\#\{(a,b)\in\mathbb{F}_{q^2}\times\mathbb{F}_{q^2}^*:\ W_f(a,b)=iq\}.$$
  By Lemma \ref{20240623lemma1}, Proposition \ref{20240623prop4}, Proposition \ref{20240623prop2} and the definition of the $\eta_i$'s, we have
  \begin{equation*}
    \begin{cases}
      \eta_{-1}+\eta_0+\eta_1+\eta_2=q^2(q^2-1), \\
      -q\eta_{-1}+q\eta_1+2q\eta_2=q^4-q^2,      \\
      q^2\eta_{-1}+q^2\eta_1+4q^2\eta_2=q^4(q^2-1),
    \end{cases}
  \end{equation*}
  and
  $$(-q-1)^3\eta_{-1}-\eta_0+(q-1)^3\eta_1+(2q-1)^3\eta_2=q^2(q^2-1)(q^3-3q^2+2).$$
  The desired result follows by solving the system of the four linear equations.
\end{proof}

\subsection{An application to coding theory} 

Linear codes play a crucial role in communication, consumer electronics, and data storage systems. Codes with appropriate parameters are particularly useful in such systems. Specifically, codes with few (Hamming) weights are advantageous for secret sharing and two-party computations. Moreover, cryptographic functions have significant applications in coding theory, highlighting a close and intriguing interplay between these two fields (see \cite{Mesnager-Handbook}).

In this subsection, we present an application of our results to coding theory. To begin, we recall some background information (for further details, see, for example, the book \cite{HufPless-2003}). 

A linear code over \( \mathbb{F}_q\) with parameters \( [n, k, d] \) is a \( k \)-dimensional $\mathbb{F}_q$-subspace of \( \mathbb{F}_q^n \) with a minimum Hamming distance of \( d \). The Hamming weight of a codeword \( c \) is defined as the number of its non-zero coordinates, while the Hamming distance between two codewords is the number of different coordinates. 

Let \( \tau(x_0, x_1, \cdots, x_{n-1}) \) denote the cyclic shift of the vector \( (x_0, x_1, \ldots, x_{n-1}) \), where each coordinate is mapped as \( i \mapsto i+1 \mod n \), resulting in the vector \( (x_{n-1}, x_0, \ldots, x_{n-2}) \). A linear \( [n, k] \) code \( \mathcal{C} \) over \( \mathbb{F}_q\) is called cyclic if \( c \in \mathcal{C} \) implies \( \tau(c) \in \mathcal{C} \). 

By identifying the vector \( (c_0, c_1, \ldots, c_{n-1}) \in \mathbb{F}_q^n \) with the polynomial \( \sum_{i=0}^{n-1} c_i x^i \in \mathbb{F}_q[x]/(x^n-1) \), a linear code \( \mathcal{C} \) of length \( n \) over \( \mathbb{F}_q\) corresponds to an $\mathbb{F}_q$-subspace of the quotient ring \(\mathbb{F}_q[x]/(x^n-1)\). Then the linear code \( \mathcal{C} \) is cyclic if and only if the corresponding subspace in \(\mathbb{F}_q[x]/(x^n-1) \) is an ideal of the ring \( \mathbb{F}_q[x]/(x^n-1) \). 

It is well known that every ideal of \( \mathbb{F}_q[x]/(x^n-1) \) is principal and can be generated by a unique monic factor of $x^n-1$. For any monic factor $g(x)$ of $x^n-1$, we will use $\mathcal{C}=\left\langle g(x)\right\rangle$ to denote the cyclic code corresponding to the principal ideal $\big(g(x)\big)$ of $\mathbb{F}_q[x]/(x^n-1)$ and call $g(x)$ the generator polynomial of $\mathcal{C}$. Moreover, the polynomial \( h(x):= (x^n-1)/g(x) \) is called parity-check polynomial of $\mathcal{C}$.  

Since cyclic codes support efficient encoding and decoding algorithms, they are widely applied in storage and communication systems. One way of constructing cyclic codes over \( \mathbb{F}_q\) with length \( n \) is to use a generator polynomial of the form
\[
\frac{x^n-1}{\gcd(S(x), x^n-1)},
\]
where \( S(x) = \sum_{i=0}^{n-1} s_i x^i \in \mathbb{F}_q[x] \). This approach has led to impressive progress in constructing cyclic codes in the last decade (see, for instance, \cite{DIT1, DSIAM, DFFA1}).

Let $\alpha$ be a primitive element of \(\mathbb{F}_{q^2}\). We now consider the ternary cyclic code \(\mathcal{C}\) of length \(q^2-1\) whose parity-check polynomial is given by \(p(x) = p_1(x)p_2(x)\), where \(p_1(x)\) and \(p_2(x)\) are the minimal polynomials of \(\alpha^{-1}\) and \(\alpha^{-(q+2)}\) over \(\mathbb{F}_3\), respectively.  According to Delsarte's theorem \cite{delsarte1975subfield}, the cyclic code \(\mathcal{C}\) can be expressed as follows:
$$\mathcal{C}=\Big\{c_{a,b}=\Big(\trace(a\alpha^{i(q+2)}+b\alpha^i)\Big)^{q^2-2}_{i=0}:\ a,b\in\mathbb{F}_{q^2}\Big\}.$$

\begin{corollary}
  The ternary cyclic code $\mathcal{C}$ has parameters $[q^2-1,4m,\frac{2q(q-2)}{3}]$. Moreover, the weight distribution of $\mathcal{C}$ is given in Table \ref{table3}.
  \begin{table}[h!]
    \centering
    \caption{The weight distribution of $\mathcal{C}$\\}
    \label{table3}
    \begin{tabular}{lll}
      \toprule
      Weight              &  & Number of codewords         \\
      \midrule
      $0$                 &  & $1$                         \\
      $\frac{2q(q-2)}{3}$ &  & $\frac{q^4-q^3-q^2+q}{6}$   \\
      $\frac{2q(q-1)}{3}$ &  & $q^3-q$                     \\
      $\frac{2q^2}{3}$    &  & $\frac{q^4-q^3+q^2+q}{2}-1$ \\
      $\frac{2q(q+1)}{3}$ &  & $\frac{q^4-q^3-q^2+q}{3}$   \\
      \bottomrule
    \end{tabular}
  \end{table}
\end{corollary}

\begin{proof}
  Since $p_1(x)$ is the minimal polynomial of the primitive element $\alpha^{-1}$ of $\mathbb{F}_{q^2}$ over $\mathbb{F}_3$, we have $\deg p_1=2m$. Moreover, we have
  \begin{align*}
    \deg p_2 & =\min\{j\in\mathbb{N}_+:\ \alpha^{-(q+2)\cdot 3^j}=\alpha^{-(q+2)}\} \\
             & =\min\{j\in\mathbb{N}_+:\ (q^2-1)\mid(q+2)(3^j-1)\},
  \end{align*}
  where $\mathbb{N}_+$ is the set of positive integers. Note that $\gcd(q+1,q+2)=\gcd(q+1,1)=1$ and $\gcd(q-1,q+2)=(q-1,3)=1$, which implies that $\gcd(q^2-1,q+2)=1$. It follows that
  $$\deg p_2=\min\{j\in\mathbb{N}_+:\ (q^2-1)\mid(3^j-1)\}=2m$$
  and thus $\deg p=4m$. Therefore, the dimension of $\mathcal{C}$ over $\mathbb{F}_3$ is $4m$.\\
  \indent For any $a,b\in\mathbb{F}_{q^2}$, we have
  \begin{align*}
    w_H(c_{a,b}) & =q^2-1-\#\{0\le i\le q^2-2:\ \trace(a\alpha^{i(q+2)}+b\alpha^i)=0\}                                                         \\
                 & =q^2-1-\#\{x\in\mathbb{F}_{q^2}^*:\ \trace(ax^{q+2}+bx)=0\}                                                                 \\
                 & =q^2-\frac{1}{3}\sum\limits_{y\in\mathbb{F}_3}\sum\limits_{x\in\mathbb{F}_{q^2}}\xi_3^{y\trace(ax^{q+2}+bx)}                \\
                 & =\frac{2q^2}{3}-\frac{1}{3}\sum\limits_{y\in\mathbb{F}_3^*}\sum\limits_{x\in\mathbb{F}_{q^2}}\xi_3^{\trace(ayx^{q+2}+byx)}.
  \end{align*}
  For any $y\in\mathbb{F}_3^*$, we have $y^{q+2}=y^q\cdot y^2=y\cdot y^2=y$, which implies that
  \begin{align*}
    w_H(c_{a,b}) & =\frac{2q^2}{3}-\frac{1}{3}\sum\limits_{y\in\mathbb{F}_3^*}\sum\limits_{x\in\mathbb{F}_{q^2}}\xi_3^{\trace\Big(a(yx)^{q+2}+b(yx)\Big)} \\
                 & =\frac{2q^2}{3}-\frac{2}{3}\sum\limits_{x\in\mathbb{F}_{q^2}}\xi_3^{\trace(ax^{q+2}+bx)}                                               \\
                 & =\frac{2q^2}{3}-\frac{2}{3}W_f(-b,a).
  \end{align*}
  It follows that for any $b\in\mathbb{F}_{q^2}$, we have
  $$w_H(c_{0,b})=\frac{2q^2}{3}-\frac{2}{3}\sum\limits_{x\in\mathbb{F}_{q^2}}\xi_3^{\trace(bx)}=\begin{cases}
      0,              & \text{if }b=0,    \\
      \frac{2q^2}{3}, & \text{if }b\ne 0.
    \end{cases}$$
  Moreover, using Theorem \ref{20240630them1}, the value distribution of $w_H(c_{a,b})$ $(a\in\mathbb{F}_{q^2}^*,\ b\in\mathbb{F}_{q^2})$ is given by
  $$w_H(c_{a,b})=\begin{cases}
      \frac{2q(q+1)}{3}, & \text{occurs}\ \frac{q^4-q^3-q^2+q}{3}\ \text{times}, \\
      \frac{2q^2}{3},    & \text{occurs}\ \frac{q^4-q^3-q^2+q}{2}\ \text{times}, \\
      \frac{2q(q-1)}{3}, & \text{occurs}\ q^3-q\ \text{times},                   \\
      \frac{2q(q-2)}{3}, & \text{occurs}\ \frac{q^4-q^3-q^2+q}{6}\ \text{times}.
    \end{cases}$$
  This completes the proof.
\end{proof}

\section{Conclusion and remarks}\label{section5}

In this paper, we investigated the differential spectrum, boomerang spectrum, and Walsh spectrum of the power function \( f(x) = x^{q+2} \) over the finite field \( \mathbb{F}_{q^2} \), where \( q = p^m \), \( p \) is an odd prime, and \( m \) is a positive integer. 

Firstly, we presented an alternative method for determining the differential spectrum of \( f \). Next, we analyzed the boomerang uniformity of \( f \) for the cases where \( q \equiv 1 \) or \( 3\ ({\rm{mod}}\ 6) \), excluding the specific case where \( p = 5 \) and \( m \) is even.

Finally, we determined the value distribution of the Walsh spectrum of \( f \) when \( p = 3 \), demonstrating that it is 4-valued. This result indicates that the power function \( f \) with \( p = 3 \) serves as a cryptographic function with a Walsh spectrum comprising only a few distinct values, which holds significant interest in the field of cryptography. Additionally, utilizing the obtained results, we determined the weight distribution of an associated cyclic code, establishing that it is a 4-weight code.

Throughout this article, we presented refined algebraic techniques for exploring the cryptographic analysis of functions over finite fields. These techniques can help explore new avenues, particularly in investigating the boomerang uniformity of a broader family of functions.

\bibliographystyle{elsarticle-harv}
\bibliography{refs_new}



\end{document}